\documentclass[11pt]{article}

\usepackage[english]{babel}
\usepackage[utf8]{inputenc}
\usepackage{RR}
\usepackage{epsfig,color}
\usepackage{amsmath,amssymb,stmaryrd,graphicx,wasysym}
\usepackage{tabularx}
\usepackage{hyperref}
\usepackage{url}
\usepackage{prettyref}
\usepackage{xcolor}
\hypersetup{
    colorlinks = true,
    urlbordercolor = {magenta},
    citecolor = {blue},
}
\usepackage{algorithmicx,algpseudocode}
\usepackage{wrapfig}
\usepackage{subfigure}
\usepackage{pgf}
\usepackage{pgfcore}
\usepackage{pgffor}
\usepackage{paralist}
\usepackage{ulem}
\usepackage{algorithm}
\usepackage{algorithmicx}

\newtheorem{definition}{Definition}
\newtheorem{lemma}{Lemma}
\newtheorem{theorem}{Theorem}

\newtheorem{example}{Example}
\newenvironment{proof}[1]{\paragraph{Proof}#1}

\begin{document}

\normalem

\RRdate{November 2020}

\RRetitle{\bf Mixed Nondeterministic-Probabilistic Interfaces\thanks{This paper was prepared in 2019 while the second author was supported by an Inria International Chair.}}
\RRtitle{\bf Interfaces Mixtes Probabilistes-Nond\'eterministes}

\RRauthor{Albert Benveniste\thanks{INRIA/IRISA, Rennes, France. 
  {Albert.Benveniste@inria.fr}} \and Kim G. Larsen\thanks{Department of Computer
  Science, Aalborg University, Aalborg, Denmark} \and Jean-Baptiste Raclet\thanks{IRIT, Universit{\'e} de Toulouse, Toulouse, France}}
\authorhead{A. Benveniste, K.\,G. Larsen, J-B. Raclet}

\newtheorem{ccomment}{Comment}
\newtheorem{counterex}{Counterexample}

\newcommand{\noexcept}{free of exception}
\newcommand{\yesexcept}{subject to exceptions}
\newcommand{\probamay}{{\mbox{$\diamond$}}}
\newcommand{\supp}{\mathbf{supp}}
\newcommand{\produces}[2]{#1\leadsto{#2}}
\newcommand{\probamust}{\Box}
\newcommand{\wait}{w}
\newcommand{\location}{l}
\newcommand{\Locations}{{L}}
\newcommand{\Id}{\mathbb{I}}
\newcommand{\interthrough}[3]{#1{\leftarrow\!\!\!\!\!\!\!\!\frac{~#2~}{}\!\!\!\!\!\!\!\!\ra}#3}
\newcommand{\ID}{\mathcal{I}}
\newcommand{\dirac}{\delta}
\newcommand{\cons}{C}
\newcommand{\compat}{\,{\bowtie}\,}
\newcommand{\noint}{\mapsto}
\newcommand{\bidirnoint}{{\leftarrow\!\mapsto}}
\newcommand{\system}{S}
\newcommand{\tsystem}{T}
\newcommand{\Systems}{\mathcal{S}}
\newcommand{\PNS}{\Systems^{\sf np}}
\newcommand{\Labels}{\Lambda}
\newcommand{\rond}{\circ}
\newcommand{\bip}{\textsc{Bip}}
\newcommand{\speeds}{\textsc{Speeds}}
\newcommand{\cesar}{\textsc{Cesar}}
\newcommand{\autosar}{\textsc{Autosar}}
\newcommand{\wmodels}{\models_w}
\newcommand{\smodels}{\models_s}
\newcommand{\Guards}{\mathcal{G}}
\newcommand{\Funct}{\mathcal{F}}
\newcommand{\zero}{\mathbf{0}}
\newcommand{\abstraction}{\alpha}
\newcommand{\concret}{\gamma}
\newcommand{\region}{\vartheta}
\newcommand{\Regions}{\Theta}
\newcommand{\val}{\nu}
\newcommand{\Vals}{\mathcal{V}}
\newcommand{\Clocks}{\mathcal{H}}
\newcommand{\wleq}{\leq_w}
\newcommand{\sleq}{\leq_s}
\newcommand{\sgeq}{\geq_s}
\newcommand{\imodels}{\models^{\mbox{\sc\scriptsize m}}}
\newcommand{\emodels}{\models^{\mbox{\sc\scriptsize e}}}
\newcommand{\rmecs}{\textsc{rmecs}}
\newcommand{\mecs}{\textsc{mecs}}
\newcommand{\eca}{\mbox{event-clock automaton}}
\newcommand{\ecaa}{\mbox{event-clock automata}}
\newcommand{\pta}{\mbox{\textsc{pta}}}
\newcommand{\mta}{\mbox{\textsc{mta}}}
\newcommand{\mdp}{\mbox{\textsc{mdp}}}
\newcommand{\pa}{\mbox{\textsc{pa}}}
\newcommand{\spa}{\mbox{\textsc{spa}}}
\newcommand{\mnp}{\mbox{\textsc{mnp}}}
\newcommand{\mmdp}{\mbox{\textsc{mmdp}}}
\newcommand{\mcs}{\mbox{\textsc{mc}}}
\newcommand{\cmcs}{\mbox{\textsc{cmc}}}
\newcommand{\iomdp}{\mbox{i/o-\textsc{mdp}}}
\newcommand{\iohmdp}{\mbox{i/o-\textsc{hmdp}}}
\newcommand{\ioa}{\mbox{i/o-automaton}}
\newcommand{\ioaa}{\mbox{i/o-automata}}
\newcommand{\cmpi}{\textsc{cmpi}}
\newcommand{\mpi}{\textsc{mpi}}
\newcommand{\mmi}{Mixed Interface}
\newcommand{\fompi}{\textsc{fompi}}
\newcommand{\simuproj}{\xi}

\newcounter{cexample}
\newenvironment{cexample}[1][]{\refstepcounter{cexample}\par\medskip
   \noindent \textbf{Counterexample~\thecexample #1} \rmfamily}{\medskip}

%

\newcommand{\except}{exception}
\newcommand{\excpt}[1]{\top_{\!{#1}}}
\newcommand{\noexcpt}[1]{#1^{[\top]}}

\newcommand{\albert}[1]{{\color{blue}\textsf{Albert:~#1}}}
\newcommand{\removalbert}[1]{}
\newcommand{\finremovalbert}[1]{}
\newcommand{\benoit}[1]{\fxnote{{\color{purple}\textsf{Beno\^{\i}t:~#1}}}}
\newcommand{\myparagraph}[1]{~\\{\vspace*{-3mm}}{\\}\noindent\emph{#1}:}
\newcommand{\regionof}[1]{\mathbb{R}\left[#1\right]}
\newcommand{\uupper}[1]{\overline{#1}}
\newcommand{\llower}[1]{\underline{#1}}
\newcommand{\ggeq}[1]{^\geq{#1}}
\newcommand{\lleq}[1]{^\leq{#1}}
\newcommand{\ecaof}[1]{\mathbb{M}\left[#1\right]}
\newcommand{\prog}[1]{{\footnotesize\textsf{#1}}}
\newcommand{\noproba}[1]{#1^\Downarrow}
\newcommand{\wemph}[1]{{\color{white}#1}}
\newcommand{\bemph}[1]{{\color{blue}#1}}
\newcommand{\remph}[1]{{\color{red}#1}}
\newcommand{\gemph}[1]{{\color{green}#1}}
\newcommand{\blemph}[1]{{\color{black}#1}}
\newcommand{\prune}[1]{\left[#1\right]}
\newcommand{\consistent}[1]{[#1]}
\newcommand{\extpre}[1]{{\it pre}_?\!\left(#1\right)}
\newcommand{\intpre}[1]{{\it pre}_!\!\left(#1\right)}
\newcommand{\intpreindex}[2]{{\it pre}_!^{#2}\!\left(#1\right)}
\newcommand{\intprestar}[1]{{\it pre}_!^*\!\left(#1\right)}
\newcommand{\wwriter}[1]{{\color{red}#1}}

\newcommand{\boxeq}[1]{\beqq\mbox{\fbox{$\displaystyle #1$}}\eeqq}
\newcommand{\modint}{\mathcal{C}}
\newcommand{\Components}{\mathcal{M}}
\newcommand{\nComponents}{\mathcal{N}}
\newcommand{\Events}{\Sigma}
\newcommand{\trace}{\sigma}
\newcommand{\ttrace}{\tau}
\newcommand{\strategy}{\varphi}
\newcommand{\reldot}{\bullet}
\newcommand{\event}{\sigma}
\newcommand{\may}{{\it may}}
\newcommand{\mmay}{{\bf may}}
\newcommand{\maynot}{{\it mustnot}}
\newcommand{\must}{{\it must}}
\newcommand{\mmust}{{\bf must}}
\newcommand{\xor}{{\bf Xor}}
\newcommand{\Illegal}{{\it Illegal}}
\newcommand{\altsimu}{\Delta}
\newcommand{\pseudoq}{/\!\!/}
\newcommand{\iin}{{\rm in}}
\newcommand{\oout}{{\rm out}}
\newcommand{\lloc}{{\rm loc}}
\newcommand{\boolout}[2]{b_{#2}^{\mbox{\tiny\sf {#1}}}}
\newcommand{\compile}[1]{{\bf comp}\!\left(#1\right)}
\newcommand{\ital}[1]{{\Large $#1$}}
\newcommand{\empir}[1]{\widehat{#1}}
\newcommand{\exit}[1]{\textit{Exit}\left(#1\right)}
\newcommand{\simu}{\rho}
\newcommand{\liftsimu}[1]{{#1}}
\newcommand{\NMPlift}[1]{\,#1^\Systems\,}
\newcommand{\explicitliftsimu}[1]{{#1}^{{\proba}}}
\newcommand{\subsimu}[1]{\subseteq^{#1}}
\newcommand{\supsimu}[1]{\supseteq^{#1}}
\newcommand{\insimu}[1]{\in^{#1}}

\newcommand{\ndashrightarrow}{{\raisebox{0.4mm}{\mbox{\tiny$/$}}\hspace*{-4.2mm}\dashrightarrow}}
\newcommand{\nlongrightarrow}{{\raisebox{0.4mm}{\mbox{\tiny$/$}}\hspace*{-4.2mm}\longrightarrow}}
\newcommand{\reset}[2]{#1_{\downarrow{#2}}}
\newcommand{\longtrans}[3]{{#1}\,\frac{\;#2}{~}\!\!\!\!\!\ra\,{#3}}
\newcommand{\longregiontrans}[3]{{#1}\,\stackrel{#2}{=\!=\!\Longrightarrow}\,{#3}}
\newcommand{\regiontrans}[3]{{#1}\,\stackrel{#2}{\Longrightarrow}\,{#3}}
\newcommand{\trans}[4]{{#1}\stackrel{#2}{\longrightarrow}_{#4}{#3}}
\newcommand{\probatrans}[3]{{#1}\stackrel{#2}{\Ra}{#3}}
\newcommand{\fullprobatrans}[5]{\probatrans{{\trans{#1}{#2}{#3}{#4}}}{}{#5}}
\newcommand{\fullprobatransmay}[5]{\probatrans{{\transmayindex{#1}{#2}{#3}{#4}}}{}{#5}}
\newcommand{\fullprobatransmust}[5]{\probatrans{{\transmustindex{#1}{#2}{#3}{#4}}}{}{#5}}
\newcommand{\startrans}[3]{{#1}\stackrel{#2}{\longrightarrow}{\!\!}^*{#3}}
\newcommand{\transacc}[3]{{#1}{\stackrel{#2}{\longrightarrow}}{#3}}
\newcommand{\transmay}[3]{{#1}{\stackrel{#2}{\dashrightarrow}}{#3}}
\newcommand{\transmust}[3]{{#1}{\stackrel{#2}{\longrightarrow}}{#3}}
\newcommand{\transmayprime}[3]{{#1}{\stackrel{#2}{\dashrightarrow'}}{#3}}
\newcommand{\transmustprime}[3]{{#1}{\stackrel{#2}{\longrightarrow'}}{#3}}
\newcommand{\transmayindex}[4]{{#1}{\stackrel{#2}{\dashrightarrow}}_{#4}\,{#3}}
\newcommand{\uptransmayindex}[4]{{#1}{\stackrel{#2}{\dashrightarrow}^{\!#4}}{#3}}
\newcommand{\ntransmayindex}[4]{{#1}\;{\stackrel{#2}{~\ndashrightarrow\!\!}}_{\;#4}{#3}}
\newcommand{\transmustindex}[4]{{#1}{\stackrel{#2}{\longrightarrow}}_{#4}{#3}}
\newcommand{\uptransmustindex}[4]{{#1}{\stackrel{#2}{\longrightarrow}^{\!#4}}{#3}}
\newcommand{\nuptransmustindex}[4]{{#1}{\stackrel{#2}{~\nlongrightarrow\!\!}^{#4}}{#3}}
\newcommand{\ntransmustindex}[4]{{#1}\;{\stackrel{#2}{~\nlongrightarrow\!\!}}_{\;#4}\,{#3}}
\newcommand{\transwhatever}[3][]{{#2}{\stackrel{#3}{\rightsquigarrow}}_{#1}}
\newcommand{\ntrans}[2]{{#1}\stackrel{#2}{\nrightarrow}}

\newcommand{\ntransmust}[3][]{{#2}\stackrel{#3}{~\nlongrightarrow\!\!}_{{#1}}}
\newcommand{\ntransmay}[3][]{{#2}\stackrel{#3}{~\ndashrightarrow\!\!}_{#1}}

\newcommand{\transindex}[4]{{#1}\stackrel{#2}{\longrightarrow}_{\!{#4}}{#3}}
\newcommand{\probatransindex}[4]{{#1}\stackrel{#2}{\longmapsto}_{\!{#4}}{#3}}
\newcommand{\weaktransindex}[4]{{#1}\frac{~{#2}~~}{}\!\!\!\mbox{\tiny \raisebox{0.5mm}{$>$}}_{\!\!_{#4}}^{\!\!{\diamond}}{#3}}
\newcommand{\maytrans}[3]{\xymatrix{{#1}\ar[r]^-{#2}_>{\diamond}&\;{#3}}}

\newcommand{\ie}{{i.e.,}}
\newcommand{\eg}{{e.g.,}}
\newcommand{\causes}[3]{{#1}\stackrel{#2}{-\!\!-\!\!\!\longrightarrow}{#3}}
\newcommand{\rhoproj}[3]{\rho^{#1}_{#2 \mid #3}}
\newcommand{\invproj}[2]{{\bf pr}^{-1}_{#1}\left(#2\right)}
\newcommand{\proj}[2]{\mathbf{Pr}_{#1}\!\left(#2\right)}
\newcommand{\before}[2]{{#1}_{<{#2}}}
\newcommand{\after}[2]{{#1}_{\geq{#2}}}
\newcommand{\elim}[2]{\left[#1\right]_{#2}}
\newcommand{\bigelim}[2]{\bigl[#1\bigr]_{#2}}
\newcommand{\fusion}[2]{[\![#1]\!]_{#2}}
\newcommand{\delsmo}{\,{=\!\!\!|}\;}
\newcommand{\para}{\mathbin{\|}}
\newcommand{\mpara}{\mathbin{\times}}
\newcommand{\cpara}{\mathbin{\otimes}}
\newcommand{\regionra}{\Rightarrow}
\newcommand{\ra}{\rightarrow}
\newcommand{\ramay}{\dashrightarrow}
\newcommand{\ramust}{\rightarrow}
\newcommand{\Ra}{\Rightarrow}
\newcommand{\La}{\Leftarrow}
\newcommand{\bp}{\mathbf{p}}
\newcommand{\bx}{\mathbf{x}}
\newcommand{\bP}{\mathbf{P}}
\newcommand{\bProbas}{{\bf \Pi}}
\newcommand{\bR}{\mathbb{R}}
\newcommand{\bN}{\mathbb{N}}
\newcommand{\bZ}{\mathbb{Z}}
\newcommand{\dom}{D}
\newcommand{\domvar}{D_{\rm var}}
\newcommand{\domcont}{D_{{\rm cont}}}
\newcommand{\domproba}{D_{{\rm proba}}}
\newcommand{\eproof}{\hfill$\Box$ \smallskip}
\newcommand{\eitemproof}{\hfill$\Box$}
\newcommand{\behav}{E}
\newcommand{\yes}{\mbox{\sc yes}}
\newcommand{\no}{\mbox{\sc no}}
\newcommand{\ko}{\mbox{\sc ko}}
\newcommand{\ok}{\mbox{\sc ok}}
\newcommand{\nnot}{\mbox{\it Neg}}
\newcommand{\kkill}{\mbox{\it Kill}}
\newcommand{\Emax}{E^\star}
\newcommand{\zzero}{{_0}}
\newcommand{\partialmap}{\rightharpoonup}
\newcommand{\bC}{{\mathbf C}}
\newcommand{\bD}{{\mathbf D}}
\newcommand{\bI}{{\mathbf I}}
\newcommand{\cE}{{\mathcal E}}
\newcommand{\cA}{{\mathcal A}}
\newcommand{\cC}{{\mathcal C}}
\newcommand{\cG}{{\mathcal G}}
\newcommand{\cP}{{\mathcal P}}
\newcommand{\UU}{{\mathcal U}}
\newcommand{\II}{{\mathcal I}}
\newcommand{\XX}{{\mathcal X}}
\newcommand{\YY}{{\mathcal Y}}
\newcommand{\LL}{{\mathcal L}}
\newcommand{\PP}{{\mathcal P}}
\newcommand{\pproba}{{\mathbb P}}
\newcommand{\qproba}{{\mathbb Q}}
\newcommand{\cS}{{\mathcal S}}
\newcommand{\bfS}{{\mathbf S}}
\newcommand{\bfP}{{\mathbf P}}
\newcommand{\cmc}{{\bfS}}
\newcommand{\mc}{{\bfP}}
\newcommand{\spec}{{\mathcal S}}
\newcommand{\pspec}{{^p\!\spec}}
\newcommand{\ZZ}{{\cal Z}}
\newcommand{\HH}{{\cal H}}
\newcommand{\OO}{{\cal O}}
\newcommand{\CC}{{\cal C}}
\newcommand{\pcontract}{\PP}
\newcommand{\cset}[1]{\chi_{_{#1}}}
\newcommand{\compset}[1]{\mathcal{M}_{_{#1}}}
\newcommand{\envset}[1]{\mathcal{E}_{_{#1}}}
\newcommand{\envtest}[1]{T_{_{#1}}^{\mbox{\sc\scriptsize e}}}
\newcommand{\comptest}[1]{T_{_{#1}}^{\mbox{\sc\scriptsize m}}}
\newcommand{\observer}[1]{\mathcal{O}_{_{#1}}}
\newcommand{\contractset}[1]{\chi_{#1}}
\newcommand{\contract}{\mathcal{C}}
\newcommand{\ncontract}{\mathcal{D}}
\newcommand{\scontract}{\Gamma}
\newcommand{\sncontract}{\Delta}
\newcommand{\readyset}{\it rs}
\newcommand{\Must}[1]{#1^{\must}}
\newcommand{\May}[1]{#1^{\may}}
\newcommand{\contractmust}{\Must{\contract}}
\newcommand{\contractmay}{\May{\contract}}
\newcommand{\ncontractmust}{\Must{\ncontract}}
\newcommand{\ncontractmay}{\May{\ncontract}}
\newcommand{\Contracts}{\bC}
\newcommand{\nContracts}{\bD}
\newcommand{\Contractsof}[2]{\bC_{#1{\cpara}#2}}
\newcommand{\contracts}{\bC}
\newcommand{\routing}{\rho}
\newcommand{\delay}{\delta}
\newcommand{\safe}{{\rm ok}}
\newcommand{\fault}{{\rm nok}}
\newcommand{\failure}{{\it fail}}
\newcommand{\alphabet}{\Sigma}
\newcommand{\pre}{{\sf pre}}
\newcommand{\preset}[1]{^{\bullet\!}{#1}}
\newcommand{\inal}[1]{#1^{\rm in}}
\newcommand{\outal}[1]{#1^{\rm out}}
\newcommand{\interf}{\II}
\newcommand{\init}{I}
\newcommand{\final}{F}
\newcommand{\concat}{\mbox{\tiny $\;\bullet\;$}}
\newcommand{\view}{\gamma}
\newcommand{\guard}{\gamma}
\newcommand{\aaction}[2]{#1/\left[#2\right]}
\newcommand{\action}{\alpha}
\newcommand{\Acts}{\mathbb{A}}
\newcommand{\pure}{{\rm pure}}
\newcommand{\Interactions}{\mathcal{L}}
\newcommand{\state}{\sigma}
\newcommand{\States}{\mathcal{S}}
\newcommand{\Expressions}{{\it Expr}}
\newcommand{\express}{{\bf E}}
\newcommand{\ExtendedStates}{\overline{\States}}
\newcommand{\triv}{{\bf Triv}}
\newcommand{\present}{{\it present}}
\newcommand{\controlled}{{\bf c}}
\newcommand{\Controlled}{{\bf C}}
\newcommand{\uncontrolled}{{\bf u}}
\newcommand{\Uncontrolled}{{\bf U}}
\newcommand{\neutral}{{\bf n}}
\newcommand{\pc}{{P_{\controlled}}}
\newcommand{\pu}{{P_{\uncontrolled}}}
\newcommand{\pn}{{P_{\neutral}}}
\newcommand{\visloc}{{\bf i}}
\newcommand{\vis}[1]{#1^{\sf vis}}
\newcommand{\loc}[1]{#1^{\it loc}}
\newcommand{\go}{{\sf go}}
\newcommand{\profile}{\pi}
\newcommand{\ports}{P}
\newcommand{\bfports}{{\bf p}}
\newcommand{\variables}{V}
\newcommand{\mayra}{\ra_{\Box}\,}
\newcommand{\mustra}{\ra_{\diamond}\,}

\newcommand{\free}{{\bf f}}
\newcommand{\rec}[2]{{#1}^{{\it rec}}_{{#2}}}

\newcommand{\cont}{{\rm cont}}
\newcommand{\Cont}{\mathcal{C}}
\newcommand{\Vars}{\mathcal{V}}
\newcommand{\Ports}{\mathcal{P}}
\newcommand{\stuttering}{\varepsilon}
\newcommand{\mtrue}{\mbox{\tt T}}
\newcommand{\mfalse}{\mbox{\tt F}}
\newcommand{\true}{\mbox{\textsc{t}}}
\newcommand{\false}{\mbox{\textsc{f}}}
\newcommand{\Lra}{~\Rightarrow~}
\newcommand{\Meu}{{\cal M}}
\newcommand{\Assoc}{{\cal I}}
\newcommand{\Jac}[1]{{\cal J}_{#1}}

\newcommand{\tr}{{\rm Tr}}
\newcommand{\bE}{{\bf E}}
\newcommand{\bA}{{\bf A}}
\newcommand{\bU}{{\bf U}}
\newcommand{\bV}{{\bf V}}
\newcommand{\diag}{{\rm diag}}
\newcommand{\proba}{\pi}
\newcommand{\Proba}{\Pi}
\newcommand{\Probas}{\mathcal{P}}

\newcommand{\hbet}{\!\!\!\!\!}
\newcommand{\beq}{\begin{eqnarray}}
\newcommand{\eeq}{\end{eqnarray}}
\newcommand{\beqq}{\begin{eqnarray*}}
\newcommand{\eeqq}{\end{eqnarray*}}
\newcommand{\bea}{\begin{array}}
\newcommand{\eea}{\end{array}}
\newcommand{\bet}{\begin{tabular}}
\newcommand{\eet}{\end{tabular}}

\newcommand{\eqdef}{\,=_{\rm def}\,}
\newcommand{\EE}{\mathcal{E}}
\newcommand{\DD}{\mathcal{D}}
\newcommand{\NN}{\mathbb{N}}
\newcommand{\cF}{{\cal F}}
\newcommand{\cN}{{\cal N}}
\newcommand{\restrict}[2]{{#2}_{\left\downarrow{#1}\right.}}
\newcommand{\extend}[2]{#2^{\uparrow{#1}}}
\newcommand{\wextend}[2]{#2^{\Uparrow{#1}}}
\newcommand{\sextend}[2]{#2^{\uparrow{#1}}}

\newcommand{\mmbox}[1]{\mbox{ {\normalsize #1} }}
\newcommand{\Tr}[1]{{\bf Tr}_{#1}}
\newcommand{\nn}{p}
\newcommand{\pp}{n}
\newcommand{\RR}{\mathbb{R}}
\newcommand{\cO}{{\cal O}}
\newcommand{\unknown}{\bot}
\newcommand{\known}{\top}
\newcommand{\intime}{{in-time}}
\newcommand{\invalue}{{in-value}}
\newcommand{\Intime}{{In-time}}
\newcommand{\Invalue}{{In-value}}

\newcommand{\todo}[1]{ \{{\bf TO DO:} \textsf{#1}\}}
\newcommand{\intersect}{\cap}
\newcommand{\union}{\cup}
\newcommand{\suchthat}{\ |\ }
\newcommand{\product}{\mathbin{||}}

\newcommand{\richcomp}{{\it RC}}
\newcommand{\rcname}{X}
\newcommand{\setsofcontracts}{\left\{\cgte\right\}}
\newcommand{\optimplementation}{\left[\comp\right]}

\newcommand{\runs}{\mathcal{R}}
\newcommand{\rmR}{{\rm R}}
\newcommand{\requ}[1]{{\rm R_{\mbox{\scriptsize #1}}}}
\newcommand{\reqq}[2]{{\rm R_{\mbox{\scriptsize #1.#2}}}}
\newcommand{\reqs}{\mathcal{C}}
\newcommand{\gtee}{\mathcal{G}}
\newcommand{\assp}{A}
\newcommand{\prom}{G}
\newcommand{\eassp}{H}
\newcommand{\eprom}{R}
\newcommand{\pcgte}{{\mathcal{C}}}
\newcommand{\cgte}{C}
\newcommand{\comp}{M}
\newcommand{\ccomp}{\mathbf{C}}
\newcommand{\env}{E}
\newcommand{\envcgte}{\cgte_\env}
\newcommand{\sys}{S}
\newcommand{\intf}{I}
\newcommand{\complem}[1]{{#1}^c}
\newcommand{\join}{{\,\sqcup\,}}
\newcommand{\meet}{\sqcap}
\newcommand{\bigmeet}{\mbox{\Large $\meet$}}
\newcommand{\implements}{\downarrow}
\newcommand{\environments}{\uparrow}
\newcommand{\maxcomp}[1]{{\comp_{#1}}}
\newcommand{\maximpl}[1]{{\comp_{#1}}}
\newcommand{\maxenv}[1]{{\env^{#1}}}
\newcommand{\maxsys}[1]{{\sys_{#1}}}
\newcommand{\compress}[1]{[#1]}
\newcommand{\compref}{\preceq}
\newcommand{\profref}{\preceq}
\newcommand{\cgref}{\preceq}
\newcommand{\intfref}{\preceq}
\newcommand{\cgequiv}{\sim}
\newcommand{\intfequiv}{\sim}
\newcommand{\sysequiv}{\equiv}
\newcommand{\extension}{\hookrightarrow}
\newcommand{\completion}{\mapsto}
\newcommand{\minr}[1]{{#1^{\reqs}}}
\newcommand{\ming}[1]{{#1^{\gtee}}}
\newcommand{\dual}[1]{{\hat{#1}}}
\newcommand{\component}{implementation}
\newcommand{\Component}{Implementation}
\newcommand{\pprime}{'}
\newcommand{\staying}{staying-in-parking}
\newcommand{\Staying}{Staying-in-parking}
\newcommand{\prefines}{\stackrel{\pi}{\preceq}}
\newcommand{\refines}{\preceq}
\newcommand{\nrefines}{\npreceq}
\newcommand{\unrefines}{\succeq}
\newcommand{\bigtimes}[1]{\mbox{\Large$\times$}_{\!#1}\,}

\newcommand{\nwleq}{\nleq_w}
\newcommand{\internal}[1]{\check{#1}}

\def\IEEEproof{\noindent\textbf{Proof:}}
\def\endIEEEproof{\hfill$\Box$}

\newtheorem{examp}{{\it Example}}
%

\newcommand{\illegal}{\iota}
\newcommand{\net}{{\cal N}}
\newcommand{\ccc}{\mbox{\textsc{abClopruv}}}

\newcommand{\Pred}{\Phi}
\newcommand{\cQ}{\mathcal{Q}}
\newcommand{\cM}{\mathcal{M}}

\newcommand{\covering}[1]{\mbox{cover}\left({#1}\right)}

\newcommand{\acc}{\mbox{Acc}}

\newcommand{\lgeq}[1]{\sim_{#1}}

\newcommand{\FO}{\mbox{FO}}


\RRabstract{
  Interface theories are powerful frameworks supporting incremental
  and compositional design of systems through refinements and
  constructs for conjunction, and parallel composition. In this report
  we present a first Interface Theory---Modal Mixed Interfaces---for
  systems exhibiting both non-determinism and randomness in their
  behaviour.  The associated component model---Mixed Markov Decision
  Processes---is also novel and subsumes both ordinary Markov
  Decision Processes and Probabilistic Automata.
	}
\RRresume{
Les th\'eories d'interfaces sont des formalismes de sp\'ecification. Elles permettent une sp\'ecification incr\'ementale gr\^ace \`a une alg\`ebre riche d'op\'erateurs tels que le raffinement, la conjonction et la composition parall\`ele. Dans ce rapport, on propose une th\'eorie d'interfaces, les Interfaces Modales Mixtes, qui permettent de sp\'ecifier des syst\`emes combinant \'etroitement des aspects probabilistes et non-d\'eterministes. Les Interfaces Modales Mixtes sont construites au-dessus du mod\`ele de composant des Automates Mixtes, qui \'etend à la fois les Processus de D\'ecision Markoviens et les Automates Probabilistes.
}
\RRmotcle{Th\'eories d'interfaces, interfaces probabilistes, syst\`emes probabilistes, syst\`emes non-d\'eterministes}
\RRkeyword{Interface theories, Probabilistic interfaces, Probabilistic systems, Nondeterministic systems}
	\RRprojet{Hycomes}  
\RCRennes
\RRNo{9372}

\makeRR
\tableofcontents
\clearpage

\section{Introduction}

Contract or Interface Theories are powerful frameworks for the
incremental and compositional design of systems. Their essence is to
handle \emph{components} (for capturing actual designs) and
\emph{contracts} or \emph{interfaces} (for capturing specifications).
At their heart sits the notion of {satisfaction}, stating that a
design suitably implements a specification.  To achieve this,
frameworks of components must be equipped with a parallel composition,
and interface theories need a richer algebra to support the
incremental and compositional design of systems, namely: refinement,
conjunction, and parallel
composition~\cite{BauerDHLLNW12,DBLP:journals/fteda/BenvenisteCNPRR18}.
Different styles of frameworks for interfaces include de
Alfaro-Henzinger \emph{Interface
  Au\-to\-ma\-ta}~\cite{DBLP:conf/emsoft/AlfaroH01} Larsen et al.
\emph{Modal Automata}~\cite{AHLNW-08} and their variants, and trace
based \emph{Assume/Guarantee
  contracts}~\cite{DBLP:conf/fmco/BenvenisteCFMPS07}.

There are several cases where the underlying class of systems involves
a mix of nondeterminism and randomness. Faults and their possible
propagation through a system are naturally modelled probabilistically,
whereas lack of knowledge of scheduling principles or incomplete
information must be modelled through by nondeterminism. Since faults
may be affected by scheduling policies, frameworks supporting the
joint handling of nondeterminism and randomness are needed.
In this paper we thus ask the following natural questions:

\begin{itemize}
\item \emph{Question 1}.\; Can we develop a
  framework for components able to blend probabilistic and
  nondeterministic behaviors in a compositional way?

\item \emph{Question 2}. \; Can we develop a theory of interfaces blending probabilistic and
  nondeterministic aspects to serve as
  specifications of the former?
\end{itemize}
Regarding Question 1 about components, Markov Decision Processes
(\mdp) \cite{Puterman} provide a natural framework for capturing
randomness.  Runs of an \mdp\ proceed as follows: from a state $q$
some action $\action$ can be selected, which brings the system into a
probabilistic state $\proba$ (a probability), from which the next
state $q'$ is drawn at random.  \mdp\ compose by synchronizing over
common actions, whereas probabilistic state-choice is made
independently.  As a first contribution of this paper we propose a
non-deterministic extension of \mdp\ called \emph{Mixed Markov
  Decision Processes} (\mmdp).  Runs of \mmdp\ proceed as follows:
from a state $q$ some action $\action$ can be selected, which brings
the system into a \emph{mixed} state $\system$ (a system blending
nondeterminism and probability, albeit with no dynamics, as
illustrated in Figure~\ref{uwedtrwcuytr}-left), from which the next
state $q'$ is drawn in a mixed nondeterministic/probabilistic way.
\mmdp\ compose by synchronizing on their common actions: from
$(q_1,q_2)$, performing action $\action$ brings the composed \mmdp\ to
mixed state $\system_1{\times}\system_2$, from which the next state
$(q'_1,q'_2)$ is drawn.  One major issue is here the definition of
mixed systems and their composition (see Figure~\ref{uwedtrwcuytr}-mid
and right, latter explained in Section~\ref{uwtrdytrd}). Our proposed \mmdp\ component model subsumes that of Probabilistic
Automata \cite{DBLP:conf/concur/LynchSV03}.
\newcommand{\zz}{$\cons_1(\omega_1,q_1)\wedge\cons_2(\omega_2,q_2)$}
\newcommand{\yy}{$(\Omega_1{\times}\Omega_2,\proba_1\otimes\proba_2)$}
\begin{figure}[ht]
\flushright{\scalebox{.9}{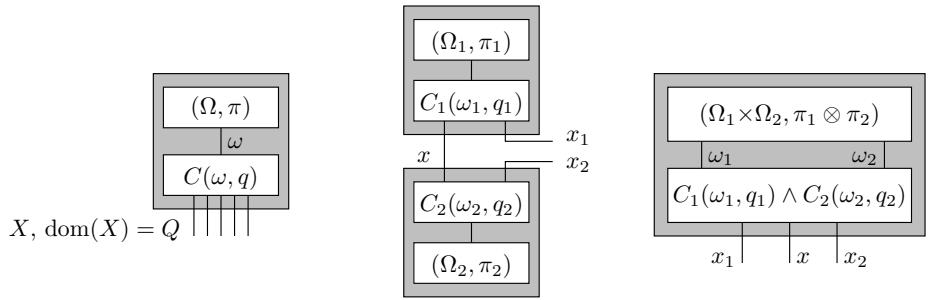}}
\caption{\sf Left: a mixed system; a probability space $(\Omega, \pi)$
  produces the private random outcome $\omega$ subject to the
  constraint $\cons(w,q)$, where $q$ is the visible value taken by a
  tuple of variables $X$. Mid: parallel composition of two mixed
  systems, the intuition: the two systems interact through their
  shared variables (here $x$). Right: the actual formal result as
  another mixed system.}
\label{uwedtrwcuytr}
\end{figure}

Regarding Question 2 about interfaces, there exist few
attempts in this direction, but no complete answer that we are aware
of.  Caillaud et
al.~\cite{DBLP:conf/qest/CaillaudDLLPW10,DBLP:journals/tcs/CaillaudDLLPW11}
propose the framework of Constrained Markov Chains (CMC) as an
extension of Interval Markov Chains \cite{DBLP:conf/lics/JonssonL91}.
By imposing constraints on transition probabilities, CMCs are a
specification theory for discrete time Markov Chains: refinement
relations are proposed, as well as constructs for conjunction and
parallel composition.  The framework is made effective by restricting
constraints on transition probabilities to be polynomial.
\emph{Abstract Probabilistic Automata}
(APA)~\cite{DBLP:conf/vmcai/DelahayeKLLPSW11,DBLP:conf/acsd/DelahayeKLLPSW11,DBLP:conf/qest/DelahayeLLPW11}
is a proposal for an interface theory for Probabilistic Automata.  APA
borrows from the CMC model the idea of setting polynomial constraints
on the transition probabilities attached to the probabilistic states,
and offer the same algebra as CMC does.

The second and major contribution of this paper is the novel framework
of Modal Mixed Interfaces (or \mmi{s} for short), which is an
interface theory for \mmdp. Our approach consists in lifting, to
\mmdp, the construction of Modal Interfaces~\cite{Raclet2011a} on top
of automata.  \mmi{s} offer the usual algebra of interface theories,
namely: satisfaction (also named implementation), refinement,
conjunction, and parallel composition.  We show that \mmi{s} extend CMC regarding satisfaction and refinement, while offering a much cleaner notion of parallel composition.

The paper is organized as follows. In Section~\ref{uwtrdytrd} we
develop the model of Mixed Systems sketched in
Figure~\ref{uwedtrwcuytr}, on top of which \mmdp\ are built in
Section~\ref{roe87wyhfgbsehlrigu} to serve as model of component. In
Section~\ref{45oiuepiruh} we show how to embed in \mmdp\ Segala's
Probabilistic Automata with nondeterministic transition
relations. Section~\ref{gw49587gtroi} introduces Mixed Modal
Interfaces as a specification framework for \mmdp\ and we show in
Section~\ref{rtuiothoguih} how to embed in it the specification
framework of Constraint Markov Chains. All proofs are deferred to
appendices.

\section{Mixed Probabilistic Nondeterministic systems}
\label{uwtrdytrd}

For $(\Omega,\proba)$ a finite or countable probability space,
$\proba$ is entirely determined by its associated \emph{weighting
  function} $w(\omega)\eqdef\proba(\{\omega\})$. By abuse of notation,
we denote by $\proba(\omega)$ the weighting function associated to
$\proba$.  Also, for a subset $W\subseteq\Omega$ such that
$\proba(W)>0$, we define the \emph{conditional probability}
$\proba(.\mid{W})$ by the formula
$\proba(V{\mid}{W})\eqdef\frac{\proba(V\cap{W})}{\proba(W)}$, which is
well defined since $\proba(W){>}0$. The \emph{support} of $\proba$,
denoted by $\supp(\proba)$, is the set of all $\omega$ such that
$\proba(\omega)>0$. Throughout this paper and unless otherwise
specified we consider only finite or countable probability
spaces.\footnote{\label{eriuyoiu} 
The restriction that $\Omega$ is at
  most countable is technically important in the above material. For
  the general case, we must abandon conditional probabilities and use
  the notion of \emph{conditional expectation,} which is defined in
  full generality. Conditional distributions require additional
  topological assumptions for their definition, and so does the notion
  of \emph{support.}}

We are now ready to define Mixed Probabilistic Nondeterministic
systems and give their semantics. This was illustrated in
Figure\,\ref{uwedtrwcuytr}.
\begin{definition}
  \label{slergiuhpiu} \it A \emph{Mixed Nondeterministic Probabilistic
    system} or \emph{Mixed System} for short is a tuple:
   $\system=((\Omega,\proba),X,\cons)$,
  where $(\Omega,\proba)$ is a probability space; $X$ is a finite set
  of variables having finite or countable domain
  $\bea{c}Q=\prod_{x\in{X}}Q_x\eea$; and
  $\cons\subseteq{\Omega\times{Q}}$ is a relation.

  A system $\system$ is called \emph{inconsistent} if
  $\proba(\exists{q}.C)=0$, otherwise it is said {consistent}.
  If $\system$ is consistent, its \emph{operational semantics}
  consists in: 
  \begin{compactenum}
  \item drawing $\omega\in\Omega$ at random according to
    $\proba(.\mid\exists{q}.C)$, and
  \item nondeterministically selecting $q\in{Q}$ such that
    $\omega\,\cons\,{q}$.
  \end{compactenum}
	This two-step procedure is denoted by
  $\produces{\system}{q}$
  and, for $\Systems$ a set of mixed systems, we write
  $\produces{\Systems}{q}$
  if $\produces{\system}{q}$ holds for some $\system\in\Systems$.
\end{definition}
In the sequel and unless otherwise specified, we only consider
consistent systems.

\begin{figure}[ht]
  \begin{center}
    \includegraphics[scale=0.3]{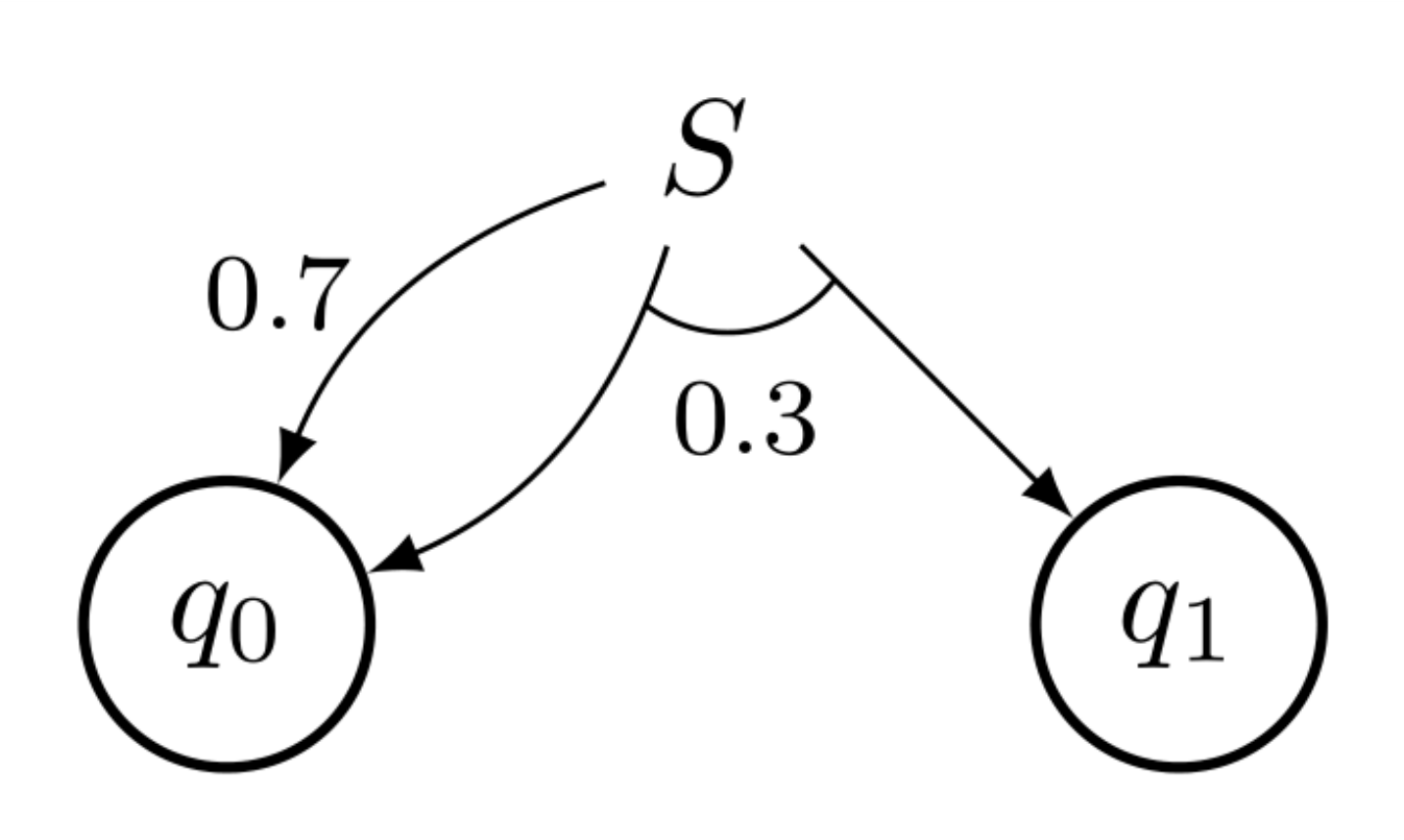}
    \end{center}
\caption{\sf Example~\ref{kghirktglrull}: a mixed system}\label{fig:ex-ms}
\end{figure}
%
\begin{example}\rm \label{kghirktglrull}
  The following mixed system $\system$ is depicted in
  Figure~\ref{fig:ex-ms}:
  \begin{compactitem}
  \item $\Omega = \{ \omega_1, \omega_2 \}$ with
    $\proba(\omega_1) = 0.7$ and $\proba(\omega_2) = 0.3$;
  \item $X = \{ x \}$ over $Q_x = \{0,1\}$. Call $q_i$
    for $i \in Q_x$ the state for which $x = i$;
  \item
    $\cons = \{ (\omega_1, q_0), (\omega_2, q_0), (\omega_2, q_1) \}$.
  \end{compactitem}
  Intuitively, $\system$ may evolve to $q_0$ with probability
  $0.7$. It may also nondeterministically evolve to $q_0$ or $q_1$
  with probability $0.3$.\eproof
\end{example}
\begin{example}\rm 
	\label{lerigtu} In a mixed system, the randoms are hidden that is, only their effect
on the visible system variables is of interest. Now suppose that
$\Omega=\Omega_1\times\Omega_2$ and the constraint $\cons$ has the
form $\cons(\omega_1,q)$. In this case, $\Omega_2$ is not needed and
can be removed, e.g., by replacing $(\Omega,\proba)$ by its marginal
$(\Omega_1,\proba_1)$, where
$\proba_1(\omega_1)=\sum_{\omega_2}\proba(\omega_1,\omega_2)$.\eproof 
\end{example}
This was just a simple case and we now discuss the operation of
compression, on top of which a notion of equivalence between systems
can be defined.  This material is borrowed
from~\cite{BenvenisteLFG95}.

\begin{definition}[compression]
  \label{lighlalegfr}
  For $\system=((\Omega,\proba),X,\cons)$ a Mixed System, we define
  the following equivalence relation on $\Omega$: \beq \omega \sim
  \omega'&\emph{iff}& \forall{q}:(\omega,q)\in\cons \Leftrightarrow
  (\omega',q)\in\cons
  \label{eoguheogihio}
  \eeq The \emph{compression} of $\system$, denoted by
$\compress{\system}=((\compress{\Omega},\compress{\proba}),X,\compress{\cons})$,
is defined as follows: $\compress{\Omega} = \Omega{/}{\sim}$ (its
elements are written $\compress{\omega}$),
$\compress{\cons}(\compress{\omega},\cdot) = {\cons}({\omega},\cdot)$
for $\omega\in\compress{\omega}$ and
$\compress{\proba}(\compress{\omega})=\sum_{\omega\in\compress{\omega}}\proba(\omega)$.
Say that $\system$ is \emph{compressed} if it coincides with its compression.
\end{definition}
Distinguishing $\omega$ and $\omega'$ is impossible if
$\omega\sim\omega'$. Compressing $\Omega$ is thus natural.
We say that two systems are equivalent if their compressed forms are
isomorphic.

\begin{definition}[equivalence] \it
  \label{oeruihytersd}
  Two compressed mixed system $\system$ and $\system'$ are called
  \emph{equivalent}, written $\system\equiv\system'$, if they possess
  identical sets of variables $X=X'$ and isomorphic operational
  semantics, i.e., if, when setting
  $\cons_\proba = \left\{ (\omega,q){\in}\cons \mid \proba(\omega){>}0
  \right\}$,
  there exists a bijective map:
  $ \varphi:\cons_\proba\mapsto\cons'_{\proba'}$
  such that, for every $(\omega,q)\in\cons_{\proba}$, we have
  $\proba(\omega)=\proba'(\omega')$ and $q=q'$, where
  $(\omega',q')\eqdef\varphi(\omega,q)$.
  Say that arbitrary systems $\system$ and $\system'$ are
  \emph{equivalent} if their compressions are equivalent.
\end{definition}

Mixed Systems are equipped with a parallel composition by intersection
in which probabilistic choices remain local and independent,
conditionally to the satisfaction of synchronization constraints.

\begin{definition}[parallel composition]
	\label{lrtguiodtrhleguip}
	For $\system_i,i=1,2$ two mixed systems, we define their
        \emph{parallel composition} $\system=\system_1\times\system_2$
        as the following Mixed System:
\vspace{-0.1cm}
        \beqq\bea{rcl}
        {X}&=&{X_1}\cup{X_2} \;,\;
	\Omega=\Omega_1\times\Omega_2 ~\emph{, and }~\proba=\proba_1\otimes\proba_2 \\ [1mm]
        \cons&=&\left\{ \left(\omega,q\right)\;\left|\;
            \omega_1\,\cons_1\,\proj{1}{{q}} \wedge
            \omega_2\,\cons_2\,\proj{2}{{q}} \right.\right\} \eea\eeqq
        \vspace{-0.1cm}
        where $\proj{i}{{q}}$ denotes the projection of the state
        ${q}$ over the variables ${X_i}$.
\end{definition}
The definition of $\cons$ expresses that the two systems must agree on
their shared variables ${X_1}\cap{X_2}$.
For the next definition, $\proj{12}{.}$ denotes the projection over
the shared variables \mbox{${X_1}\cap{X_2}$}.  We write
$q_1\compat{q_2}$ and say that $q_1$ and $q_2$ are \emph{compatible}
if $\proj{12}{{q_1}}=\proj{12}{{q_2}}$. If $q_1\compat{q_2}$, we
define the \emph{join} $q_1{\join}q_2$ as the unique $q$ projecting
over $q_1$ and $q_2$. Using this notation, $\cons$ in
Definition~\ref{lrtguiodtrhleguip} rewrites
\begin{equation}
  \cons=\left\{
    (\omega,q_1\join{q_2})
    \mid
    q_1\compat{q_2} \,\wedge\, \omega_1\cons_1{q_1} \,\wedge\, \omega_2\cons_2{q_2}
  \right\}
\label{usdqwcdur}
\end{equation}
Observe that the composition of two consistent systems may be
inconsistent.
\begin{lemma}
  \label{wjdetyfuuy}
  For mixed systems, equivalence is a congruence, i.e.,
  $\system_i\equiv\system'_i$ for $i=1,2$ implies
  $\system_1\times\system_2\equiv\system'_1\times\system'_2$.
\end{lemma}
\begin{proof}
	 See Appendix~\ref{riuygfttyiohjih}.\eproof
\end{proof}

Let $\Systems(X)$ denote the collection of all mixed systems having $X$ as
set of variables.
\begin{definition}[lifting relations]
  \label{hrgfuihsk} Relation
  {$\NMPlift{\simu}\subseteq\Systems(X_1)\times\Systems(X_2)$} is called the \emph{lifting} of relation $\simu\subseteq{Q_1}{\times}{Q_2}$ if
  there exists a \emph{weighting function} 
	\linebreak
  \mbox{$w:\Omega_1{\times}\Omega_2\ra[0,1]$} such that: 
\begin{enumerate}
\item \label{sggouigh} For every triple $(\omega_1,\omega_2;q_1)$ such
  that $w(\omega_1,\omega_2)>0$ and $\omega_1\,\cons_1\,{q_1}$, there
  exists $q_2$ such that $\omega_2\,\cons_2\,{q_2}$, and
  $q_1\,\simu\,{q_2}$;
\item \label{leiurlyui}
  $\sum_{\omega_2}w(\omega_1,\omega_2)=\proba_1(\omega_1)$ and
  \ $\sum_{\omega_1}w(\omega_1,\omega_2)=\proba_2(\omega_2)$.
\end{enumerate}
\end{definition}
Note the existential quantifier in Condition~\ref{sggouigh}.  By
Condition~\ref{leiurlyui}, $w$ induces a probability on
$\Omega_1\times\Omega_2$. We write $\system_1\NMPlift{\simu}\system_2$
to mean $(S_1,S_2)\in\NMPlift{\simu}$.
\begin{example}\label{ex:lift} \rm
  Consider the mixed systems $\system_1$ and $\system_2$ depicted in
  Figure~\ref{fig:lift}.
  We can lift the relation $\simu$ such that
  $\simu = \{ (q_{10}, q_{20}), (q_{11}, q_{20}), (q_{11}, q_{21}) \}$
  and see that $\system_1\NMPlift{\simu}\system_2$ by considering the
  weighting function shown in red.
  However, the relation $\simu'$ such that
  $\simu' = \{ (q_{10}, q_{20}), (q_{11}, q_{21}) \}$ cannot be lift
  as a witness $w$ does not exist.\eproof
\end{example}

\begin{figure}[ht]
  \begin{center}
      \includegraphics[scale=0.3]{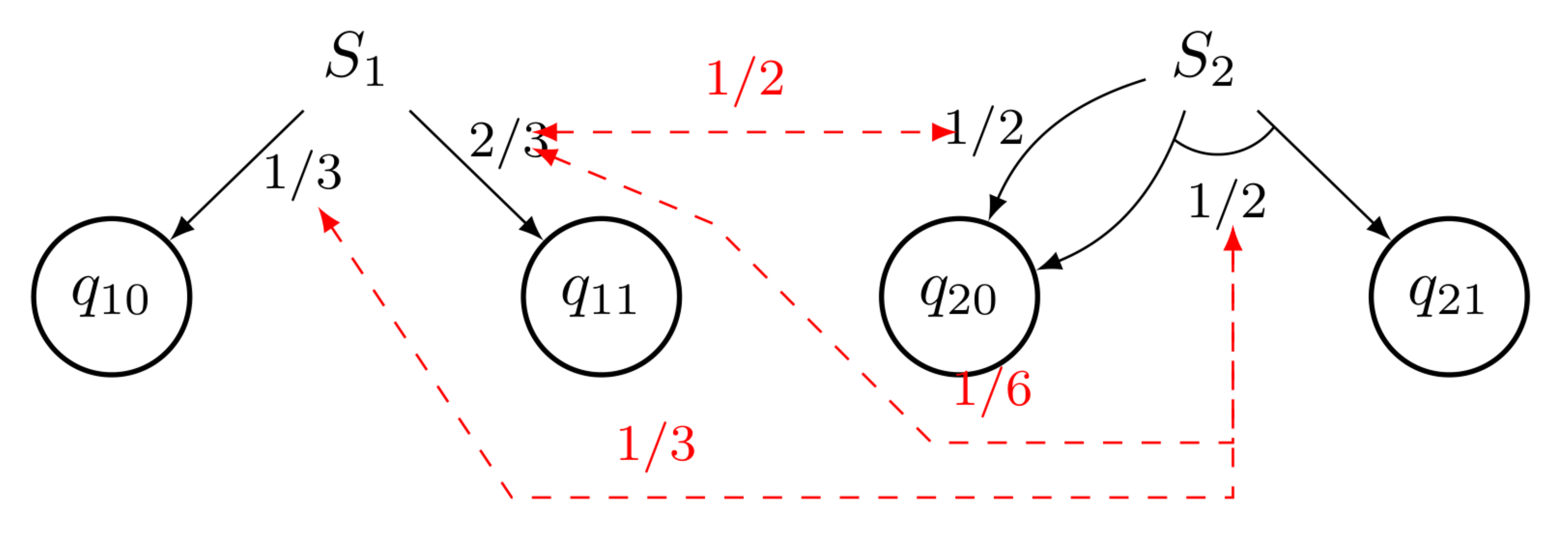}
  \end{center}
\caption{\sf Example~\ref{ex:lift}: lifted relation}\label{fig:lift}
\end{figure}

\begin{lemma}
  \label{egfuioehrpo} $\system_1\NMPlift{\simu}\system_2$ and
  $\system'_1{\equiv}\system_1$ together imply
  $\system'_1\NMPlift{\simu}\system_2$.
\end{lemma}
\begin{proof}
See Appendix~\ref{elrgfuilyu}.\eproof 
\end{proof}
Lemma~\ref{egfuioehrpo} expresses that mixed system equivalence is also a
congruence with respect to the lifting of relations. 
\begin{definition} \label{gothiioaweukdfyg} Given
  $\simu \subseteq Q_1 \times Q_2$ and
  $\Systems_i\subseteq\Systems(Q_i)$ for $i = 1,2$, define
$$\Systems_1\,\subsimu{\simu}\,\Systems_2 \emph{ iff }\forall\system_1{\in}\Systems_1, \exists\system_2{\in}\Systems_2: \system_1\NMPlift{\simu}\system_2\,,$$   
and define $\Systems_1\,\supsimu{\simu}\,\Systems_2$ as being
$\Systems_2\,\subsimu{\widetilde{\simu}}\,\Systems_1$, where
${\widetilde{\simu}}$ denotes the transpose of ${{\simu}}$. Write
$\system\,\in^{\simu}\,\Systems$ to mean
$\{\system\}\,\subsimu{\simu}\,\Systems$ and
$\Systems\,\ni^{\simu}\,\system$ to mean
$\Systems\,\supsimu{\simu}\,\{\system\}$.

\end{definition}
For $Q_1$, $Q_2$, and $Q_3$ three finite or countable sets, and
\mbox{$\simu_{12}\subseteq{Q_1}{\times}{Q_2}$} and
\mbox{$\simu_{23}\subseteq{Q_2}{\times}{Q_3}$} two relations, define:
  \vspace{-0.3cm}
\beq
\simu_{12}\reldot\simu_{23} &\eqdef& \proj{Q_1{\times}Q_3}{\simu_{12}{\wedge}\simu_{23}}
\label{rbgjgnwoun}
\eeq
 that is,
$\simu_{12}{\reldot}\simu_{23} \subseteq Q_1 {\times} Q_3$ and
$q_1(\simu_{12}{\reldot}\simu_{23})q_3$ iff $q_1\,\simu_{12}\,q_2$ and
$q_2\,\simu_{23}\,q_3$ for some $q_2 {\in} Q_2$.
\begin{lemma} \label {uweygfkiutygf}
We have $\NMPlift{(\simu_{12}\reldot\simu_{23})} = \NMPlift{\simu_{12}}\reldot\NMPlift{\simu_{23}}$ ~and~ $
\subsimu{\simu_{12}\reldot\simu_{23}} = \subsimu{\simu_{12}}\reldot\subsimu{\simu_{23}}$.
\end{lemma}
\begin{proof}
	See Appendix~\ref{guihepioru}.
\end{proof}
\paragraph{The set algebra of mixed systems.}
\label{egfuehlgui}
We have introduced in Definition~\ref{oeruihytersd} the notion of
equivalence $\equiv$ for mixed systems. Lemmas~\ref{wjdetyfuuy}
and~\ref{egfuioehrpo} show that this equivalence is a congruence with
respect to both mixed systems composition and the lifting of relations
from state spaces to mixed systems.  We now define the set algebra
induced by this equivalence. For $\Systems$ and $\Systems'$ two sets
of mixed systems:
\beq\bea{rcl} \Systems \subseteq \Systems'
&\mbox{ iff }&
\forall\system{\in}\Systems,\exists\system'{\in}\Systems':
\system'\equiv\system
\\ [1mm]
\Systems = \Systems' &\mbox{ iff }& \Systems \subseteq \Systems' ~\mbox{ and }~
\Systems' \subseteq \Systems
\\ [1mm]
\Systems_1 \cap \Systems_2 &=& \bigcup\left\{\Systems\mid\Systems\subseteq\Systems_1 \mbox{ and } \Systems\subseteq\Systems_2\right\}
\eea
\label{egieofuiohi}
\eeq Thanks to this redefinition, we shall freely use the usual set
theoretic notations for sets of mixed systems.

\section{{Mixed Markov Decision Processes}}
\label {roe87wyhfgbsehlrigu}

Probabilistic automata have been introduced in~\cite{SegalaL94} for
the study of randomization in concurrency theory. They are labeled
transitions systems where transitions are from states not to a single
target state but to a target state determined by a probability
measure. \emph{Markov Decision Processes}~\cite{BaierK00,Derman70}
exist in mathematics for quite some time. They correspond to
\emph{deterministic} probabilistic automata in the following sense:
from each state, each action identifies a unique probability
measure. In this paper we consider extensions of MDP in which the
target of a transition is a mixed probabilistic/nondeterministic
system as defined in Section~\ref{uwtrdytrd}:

\begin{definition}[\mmdp] \label{def-iohmdp} A \emph{Mixed Markov
    Decision Process} (\mmdp) is a tuple
  $M=({\alphabet},{X},r_0,\ra)$, where:
\begin{itemize}
\item ${\alphabet}$ is a finite alphabet of \emph{actions};
\item $X$ is a finite set of variables having finite or countable
  domain $R{=}\prod_{x\in{X}}R_x$, and $r_0 \in R$ is the
  \emph{initial} state;
\item $\ra\;\subseteq\;R \times \alphabet \times \Systems(X)$ is the
  \emph{transition relation}; we write~
  $\trans{r}{\action}{\system}{}$ $($or~
  $\trans{r}{\action}{\system}{M})$ to mean
  $(r,\action,\system) \in \ra$, and $\trans{r}{\action}{}{}$ if
  $\trans{r}{\action}{\system}{}$.
\end{itemize}
Let $\Systems(X)$ be the set of all mixed systems $\system$ over $X$,
possibly inconsistent.
We require that $M$ shall be \emph{deterministic}: for any pair
$(r,\action)\in{R}\times\alphabet$, $\trans{r}{\action}{\system}{}$
and $\trans{r}{\action}{\system'}{}$ implies $\system = \system'$. $M$
is said to be \emph{live} if all its transitions target consistent
systems.
\end{definition}
A \emph{run} $\trace$ of $M$ is a finite or infinite sequence of states $r_0,r_1,r_2,\dots$ starting from initial state $r_0$ and then
progressing by a sequence of steps of the form
$\fullprobatrans{r_k}{\action}{\system}{}{r_{k+1}} \,,
$
where $\produces{\system}{r'}$ is the operational semantics of system $\system$ following Definition~\ref{slergiuhpiu}.
%
 \begin{figure}[ht]
	\begin{center}
  \includegraphics[scale=0.3]{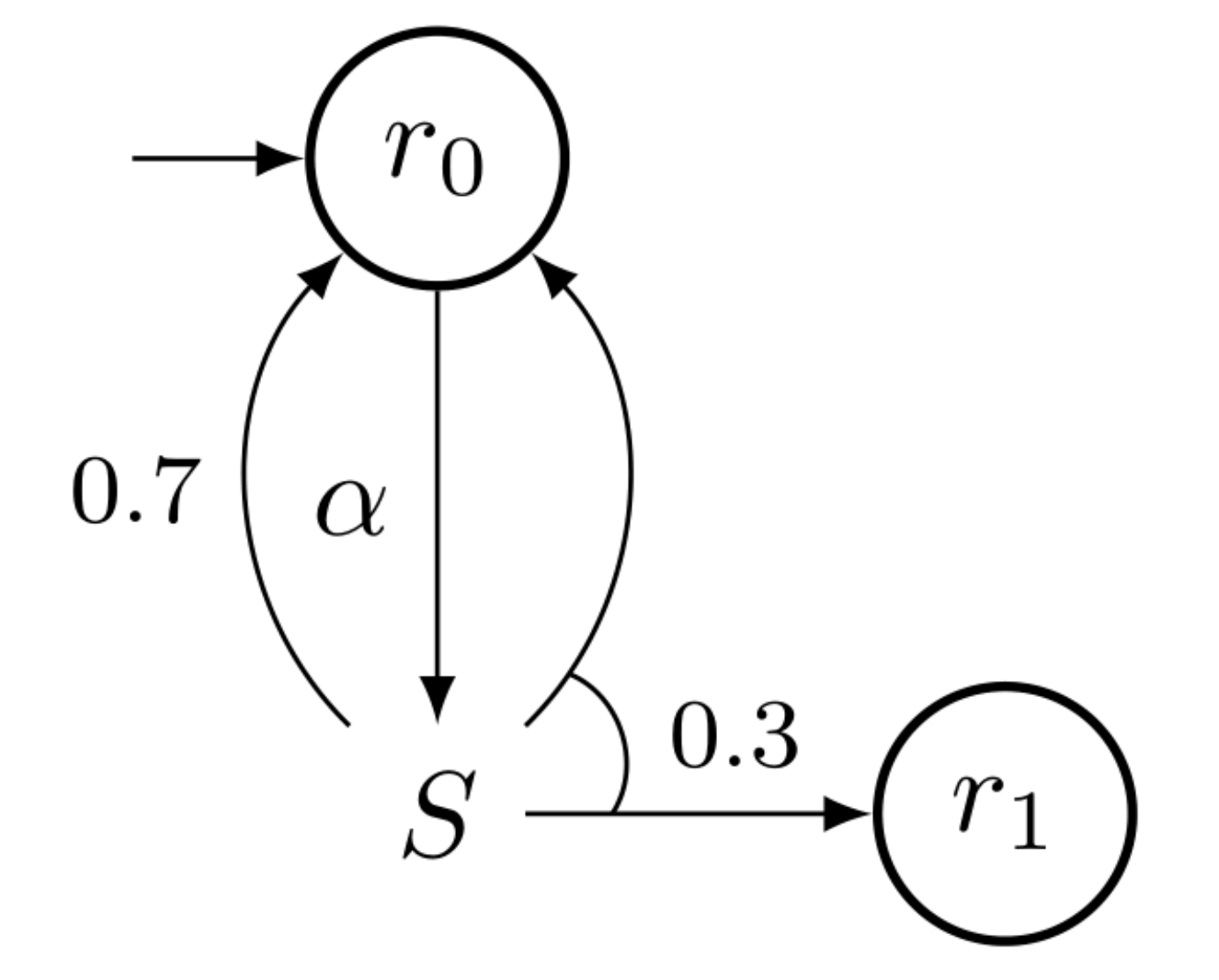}
  \caption{\sf Example of an \mmdp\ ($S$ is the mixed system of
   Figure~\ref{fig:ex-ms})}\label{fig:ex-MMDP}
 \end{center}
\end{figure}
%
 \begin{example}\rm
   An \mmdp\ is depicted in Figure~\ref{fig:ex-MMDP} with
   $\alphabet = \{ \action \}$; it has only one transition
   $(r_0, \action, S)$ where $S$ is the mixed system of
   Figure~\ref{fig:ex-ms}.\eproof
 \end{example}
\begin{definition}[simulation]\label{def:simulation}
\it Given two \mmdp\
  $M_1$ and $M_2$ over $\alphabet$, $M_2$ \emph{simulates} $M_1$,
  written $M_1 \leq M_2$, if there exists a relation $\leq \
  \subseteq \, R_1 \times R_2$ such that:
\begin{compactitem}
\item $r_{0,1} \leq r_{0,2}$\ and,
\item for $r_1 \leq r_2$ and for each transition
$\transindex{r_1}{\action}{\system_1}{M_1}$,
  there exists a transition
  $\transindex{r_2}{\action}{\system_2}{M_2}$
  such that
  $\system_1~\NMPlift{\leq}~\system_2$.
\end{compactitem}
$M_1$ and $M_2$ are called \emph{simulation equivalent} if they simulate each other.
\end{definition}
The composition of two \mmdp\ having identical alphabets is introduced next. A transition labeled $\action$ is available in the
product if and only if the components are ready to do simultaneously a
transition labeled $\action$.
\begin{definition}[composition]\label{rtghltughtui} \it
  For $M_1$ and $M_2$ two \mmdp\ having identical alphabet $\alphabet$
  and compatible initial states $r_{0,1}{\compat}r_{0,2}$, their
  \emph{composition} $M_1\mpara{M_2}$ has alphabet $\alphabet$, set of
  variables $X_1{\cup}X_2$, and initial state
  $r_{0,1}{\join}r_{0,2}$. Its transition relation is the minimal
  relation satisfying:
  \beqq
\trans{r_i}{\action}{\system_i}{M_i} \emph{ for }i=1,2  \emph{ and }
r_1\compat{r_2}
&\implies& \trans{r_1\join{r_2}}{\action}{\system_1 {\times} \system_2}{M}
\eeqq
where $\system_1 {\times} \system_2$ has been defined in Definition~$\ref{lrtguiodtrhleguip}$.
\end{definition}
Parallel composition preserves simulation:
\begin{lemma}
	\label{glrtukghtrllsdukfg}
	Let $M_i,i=1,2$ be two \mmdp\ and let $M'_i\leq{M_i},i=1,2$. Then, we have $M'_1\times{M'_2}\leq{M_1}\times{M_2}$.
\end{lemma}
\begin{proof}
	See Appendix~\ref{wroifuwgopi}.\eproof
\end{proof}

\section{Link to Probabilistic Automata}
\label{45oiuepiruh}

Probabilistic Automata (\pa)~\cite{DBLP:conf/concur/LynchSV03} are a
nondeterministic extension of \mdp{s}. We show here that \mmdp\ can
capture this nondeterminism by making use of the nondeterminism
involved in mixed systems. We discuss here the version of \pa\ with no
consideration of internal actions.
\begin{definition}
  A Probabilistic Automaton \pa\ is a tuple $P=(\alphabet,Q,q_0,\ra)$,
  where $\alphabet$ is the finite alphabet of actions, $Q$ is a finite
  state space, $q_0{\in}Q$ is the initial state, and
  $\ra\,\subseteq\,{Q}{\times}\alphabet{\times}\Probas(Q)$ is a
  probabilistic transition relation where $\Probas(Q)$ is the set of
  all probability distributions over $Q$.
\end{definition}
{The operational semantics of $P$ is as follows:} if $P$ is in state
$q{\in}Q$, performing $\action{\in}\alphabet$ leads to some target set
of probability distributions over $Q$, of which one is selected,
nondeterministically, and then used to draw the next state $q'$.
We can reinterpret this operational semantics as follows: performing
$\action{\in}\alphabet$ while being in state $q{\in}Q$ leads to the
same target set of probability distributions over $Q$, that we use
differently. We form the direct product of all distributions belonging
to the target set and we perform one trial according to this
distribution, i.e., we perform independent random trials for all
probabilities belonging to the target set. This yields a tuple of
candidate values for the next state, of which we select one,
nondeterministically.

Clearly, these two operational semantics produce identical
outcomes. Now, the latter is the operational semantics of the
\mmdp\ \mbox{$M_P=(\alphabet,\xi,q_0,\ra_P)$}, defined as follows:
$\alphabet$ is as before, $\xi$ is the system variable with domain
$Q$, $q_0$ is as before, and $\ra_P$ is the transition relation
defined as follows: $\ra_P$ maps a pair
$(q,\action)\in{Q}{\times}\alphabet$ to the mixed system
\mbox{$S=((\Omega,\Proba),\xi,\cons)$}
defined as follows. Let $n$ be the cardinality of the set
  $\{\proba\mid(q,\action,\proba)\in\ra\}$. Take for $\Omega$ the product of
  $n$ copies of $Q$, so that $\omega$ is an $n$-tuple of states:
  $\omega=(q_1,\dots,q_n)$. Take for $\Proba$ the product of all probabilities belonging to set $\{\proba\mid(q,\action,\proba)\in\ra\}$. Finally,
	 $(\omega,q)\in\cons$ if and only if $q\in\{q_1,\dots,q_n\}$.
The following theorem holds, for which the definitions of simulation
and composition of \pa\ are available
in~\cite{DBLP:conf/concur/LynchSV03}:
\begin{theorem}
  \label{erlgfuierhlpiu}
  Let $P_1, P_2$ be two \pa\ and $M_{P_1}, M_{P_2}$ be the
  corresponding \mmdp. The mapping $P\ra{M_P}$ preserves both simulation and product: $P_1\leq{P_2}$ if and only
  if $M_{P_1}\leq M_{P_2}$, and $M_{P_1\times{P_2}}$ and
  $M_{P_1}\times M_{P_2}$ are simulation equivalent.
\end{theorem}
A reverse mapping also exists. The \pa\ associated to the \mmdp\ of  Fig.\,\ref{fig:ex-MMDP} is easily guessed: performing $\action$ leads to the family of two probability spaces over $R$: $(R,\proba_1)$ where $\proba_1(r_0)=1$ and $(R,\proba_2)$ where $\proba_2(r_0)=0.7$ and $\proba_2(r_1)=0.3$. Theorem~\ref{erlgfuierhlpiu} holds for this inverse mapping as well.
So, what is the point in preferring \mmdp? The rich algebra developed in Section~\ref{uwtrdytrd} (with the two key notions of compression and lifting) is essential in supporting a flexible notion of parallel composition. In particular, when extending \pa\ with labeling using sets of atomic propositions (AP), it is required, for the parallel composition to be defined, that the two sets are disjoint. Our \mmdp\ offer the expressive power of AP-labeling without setting any restriction on the parallel composition. See Section~\ref{rtuiothoguih} for a detailed study of the same issue, for Constraint Markov Chains.

\section{{Modal Mixed Interfaces}}
  \label {gw49587gtroi}

  In this section we develop the first part of our agenda, namely a
  framework of \emph{Modal Mixed Interfaces} (or \emph{\mmi{s}} for
  short) which allow to specify sets of \mmdp\ called the
  \emph{models} of the interface. Note that in this section sets of
  probabilities associated to \mmi{s} are manipulated by not paying
  attention to effectiveness. \mmi{s} extend to a mixed
  probabilistic-nondeterministic setting the formalism of Modal
  Specifications~\cite{KT88,Larsen89,AHLNW-08}.  In this paper we
  develop our framework for the case of a fixed alphabet $\alphabet$
  of actions. Following~\cite{Raclet2011a}, alphabet extension
  techniques allow to handle the general case.

\paragraph{Definition and Semantics.}

  For $X$ a finite set of variables, $\Systems(X)$ denotes the class
  of all mixed systems $\system$ over $X$ and we call \emph{mixed
    state} a subset $\Systems\subseteq\Systems(X)$.
  \begin{definition} \label{def-empi}  A \emph{\mmi} is
    defined as a tuple \mbox{$\contract=({\alphabet},X,q_0,\ramust,\ramay)$},
    where:
\begin{compactitem}
\item ${\alphabet}$ is the finite alphabet of actions;
\item $X$ is a finite set of variables having finite domain $Q\eqdef\prod_{x\in{X}}Q_x$:
\item $q_0$  is the
  \emph{initial state} (we do \emph{not} require that $q_0 \in Q$);
\item
\mbox{$\ramust,\ramay\;\subseteq\;Q\times\alphabet\times{2^{\Systems(X)}}$}  are the \emph{must} and \emph{may} transition relations.
\end{compactitem}
We require that $\contract$ is \emph{deterministic} in the
following sense: for any pair $(q,\action)\in{Q}\times\alphabet$,
$({q},{\action},{\Systems})\in\ramust$ and
$({q},{\action},{\Systems'})\in\ramust$ imply
${{\Systems}}={{\Systems'}}$, and similarly for $\ramay$.
\end{definition}
We write
  $\transmust{q}{\action}{\Systems}$ to mean
  $(q,\action,\Systems)\in\ramust$; $\transmay{q}{\action}{\Systems}$ is defined similarly.
We write $\ntransmust{q}{\action}{}$ if there exists no mixed state $\Systems$ such that  $\transmust{q}{\action}{\Systems}$;
$\ntransmay{q}{\action}{}$ is defined similarly. Finally, we write $\probatrans{\Systems}{}{q'}$ to mean that $\probatrans{\system}{}{q'}$ holds for some $\system\in{\Systems}$.
Note that $q_0\not\in{Q}$ will typically arise when the subset ${Q}$
of states is empty; it will be useful to model \emph{unsatisfiable}
interfaces.
Whenever convenient, we shall write $\Systems^\probamust$ and $\Systems^\probamay$ when referring to mixed states targeted by \must\ and \may\ transitions, respectively.
\begin{example}\rm \label{rleuglkkkjh}
  The following \mmi\ is depicted in Figure~\ref{fig:ex-ref}-right:
  \begin{compactitem}
  \item $\alphabet = \{ \action \}$;
  \item $X = \{ x \}$ over $Q_x = \{ 0, 1 \}$ with $x = 0$ in
    $q_{0,1}$ and $x = 1$ in $q_{0,1}$;
  \item $\transmay{q_{0,2}}{\action}{\{ S_1, S_2\}}$ with $S_1$ and
    $S_2$ two mixed systems.
  \end{compactitem}
  Note that in the \mmi\ of Figure~\ref{fig:ex-ref}-left, we have
  $\transmay{q_{0,1}}{\action}{\{ S\}}$ and
  $\transmust{q_{0,1}}{\action}{\{ S\}}$ but only the plain arrow
  corresponding to the must transition is depicted in order to lighten
  the figure.\eproof
\end{example}
The intuitive semantics is the following: a \emph{must} transition
labeled by $\action$ must be available in any model with an associated
system $\system$ selected from $\Systems^\probamust$ and then a next
state $q'$ is selected according to the operational semantics of
$\system$. The same holds for a \emph{may} transition except that in
this case, the occurrence of the action is allowed but not required and
the selected system belongs to $\Systems^\probamay$.

We now formally define the notion of \emph{model} of a \mmi\ over
$\alphabet$ in terms of \mmdp\ over the same alphabet; we make use of
Definition~\ref{slergiuhpiu} for the notion of \emph{consistent
  system}, Definition~{\ref{gothiioaweukdfyg}} for the meaning of
$\in^{\models}$ and Definition~{\ref{def-iohmdp}} for \emph{live}
\mmdp:
\begin{definition}[satisfaction] \label{def-model} For $\contract$
  a \mmi\ such that $q_0\in{Q}$ and $M$ a live \mmdp, a relation
  \mbox{$\models~\subseteq~R \times {Q}$} is a \emph{satisfaction
    relation} iff, for any $(r,q)$ such that $r \models q$, the
  following holds:
%
\beq
\left.
\bea{l}\emph{only \emph{may} transitions} \\ \emph{of $\contract$ are allowed
  for $M$}\eea
	\right\}
\forall\action:
\transindex{r}{\action}{\system_M}{M}
\Ra\left[
\transmayindex{q}{\action}{\Systems^\probamay}{\contract} \emph{ and }  {\system_M}\in^{\models}{\Systems^\probamay}\right]
\label{eq-modele-may}
\\
\left.\bea{l}\emph{\emph{must} transitions of $\contract$} \\ \emph{are mandatory
  for $M$}\eea\right\}
\forall\action:
\transmustindex{q}{\action}{\Systems^\probamust}{\contract}
\Ra
\left[\transindex{r}{\action}{\system_M}{M} \emph{ and }  {\system_M}\in^{\models}{\Systems^\probamust}\right]
\label{eq-modele-must}
\eeq
$M$ is a model of $\contract$, written $M \models \contract$, if
$r_0 \models q_0$. A \mmi\ $\contract$ such that $q_0\notin{Q}$ does
not admit any model.
\end{definition}
The set of models of a Mixed Interface is closed under the simulation
equivalence of Definition~\ref{def:simulation}. Observe moreover that
the condition (\ref{eq-modele-must}) makes only sense because we
consider deterministic interfaces, since the system $\system_M$
reached by performing action $\action$ is unique in this case.

Note that, by definition, $r \models q$ induces constraints on the set
of systems associated to the \must\ and \may\ transitions stemming
from $q$. More precisely, for any $\action$ and $\Systems^\probamay$
and $\Systems^\probamust$ as in (\ref{eq-modele-may}) and
(\ref{eq-modele-must}), the intersection
$\Systems^\probamust\cap\Systems^\probamay$ necessarily contains at
least one consistent system. In this statement and in the sequel, we
stress that the set algebra over sets of Mixed Systems is the one
defined in (\ref{egieofuiohi}).
\begin{definition}
  \label{rtguiohkhtfiyt} A state $q$ is called \emph{inconsistent} if
  $\transmustindex{q}{\action}{\Systems^\probamust}{\contract}$, and
  either \mbox{$\ntransmayindex{q}{\action}{}{\contract}$}, or
  $\transmayindex{q}{\action}{\Systems^\probamay}{\contract}$ but the
  intersection $\Systems^\probamust\cap\Systems^\probamay$ contains no
  consistent system.
\end{definition}
The subset of consistent systems of
$\Systems^\probamust\cap\Systems^\probamay$ entirely specifies the set
of models of the considered \mmi. This leads to the operation of
pruning that we introduce next.  The \emph{pruning} of
$\contract$, written $\prune{\contract}$, is obtained as follows:
\begin{compactenum}
\item Let $\contract'$ the \mmi\ obtained from $\contract$ by thinning
  $\Systems^\probamust$ down to the intersection
  $\Systems^\probamust\cap\Systems^\probamay$;
\item Apply repeatedly the following transformation until fixed point,
  with initial value $k=0$ and $\contract_0={\contract'}$:
\begin{compactenum}
\item \label{05t89hpouihgpsdrui} Let $Q_{k,{\rm incon}}$ be the set of
  states $q$ of $\contract_k$ such that all inconsistent states of the
  state space $Q_k$ and set $Q_{k+1}=Q_k-Q_{k,{\rm incon}}$; by
  construction, replacing $Q_k$ by $Q_{k+1}$ does not modify the set
  of models of $\contract$;
  \item
  Performing
  this step may create new inconsistent states, however; and, thus,  we set $k\leftarrow{k+1}$ and return to
  step~\ref{05t89hpouihgpsdrui}.
\end{compactenum}
\end{compactenum}
Let $\prune{\contract}$ be the \mmi\ obtained at fixed point.
\begin{lemma}\label{lemma-cleaning}
  By construction, $\prune{\contract}$ and ${\contract}$ possess
  identical sets of models.
\end{lemma}
\begin{proof}
See Appendix~\ref{erpguioehpguio}.\eproof
\end{proof}
Note that by considering that \mmi{s} have finite sets of states, the
pruning procedure is terminating.
A \mmi\  $\contract$ is called \emph{inconsistent} iff it has no model,
i.e. iff the initial state $q_0$ does not belong to the set of states
of $\prune{\contract}$.
Unless otherwise specified, we assume in the sequel that:
\beq
\mbox{
\begin{minipage}{6cm}
	 {Pruning has been applied to every considered \mmi: $\prune{\contract}={\contract}$.}
\end{minipage}
}
\label {o59tgoehyrtiutid}
\eeq
\paragraph{Refinement.}
We now consider refinement which aims at comparing interfaces at
different stages of their design. Intuitively, it allows to check if
an interface is a more detailed version of an initial one. More
precisely, refining an interface amounts to exclude some potential
models from its set of models.
\begin{definition}[modal refinement]\label{def-refinement} \it
  Let $\contract_i,i=1,2$ be two {\mmi{s}} over $\alphabet$, a
  relation $\refines \; \subseteq \; {Q_1} \times {Q_2}$ is a
  \emph{modal refinement} iff, for all $(q_1,q_2)$ such that
  $q_1\refines{q_2}$ and for every $\action\in\alphabet$:
\vspace{-0.2cm}
  \beq
\bea{rcr}
\transmayindex{q_1}{\action}{\Systems^\probamay_1}{1}
&\  \Ra \  &
\transmayindex{q_2}{\action}{\Systems^\probamay_2}{2}
\mbox{ and } \ {\Systems^\probamay_1}\subseteq^\refines\,{\Systems^\probamay_2}
\\
\transmustindex{q_2}{\action}{\Systems^\probamust_2}{2}
&\  \Ra \  &
\transmustindex{q_1}{\action}{\Systems^\probamust_1}{1}
\mbox{ and } \ {\Systems^\probamust_1}\subseteq^\refines\,{\Systems^\probamust_2}
\eea
\label {507tho5u0857}
\eeq
Say that $\contract_1$ is a \emph{modal refinement} of $\contract_2$,
written $\contract_1\refines\contract_2$, if for $q_{0,1}\in Q_1$ and
$q_{0,2}\in Q_2$, we have $q_{0,1}\refines{q_{0,2}}$.

\end{definition}
\begin{figure}[ht]
  \begin{minipage}[c]{.48\linewidth}
   \begin{center}
       \includegraphics[scale=0.3]{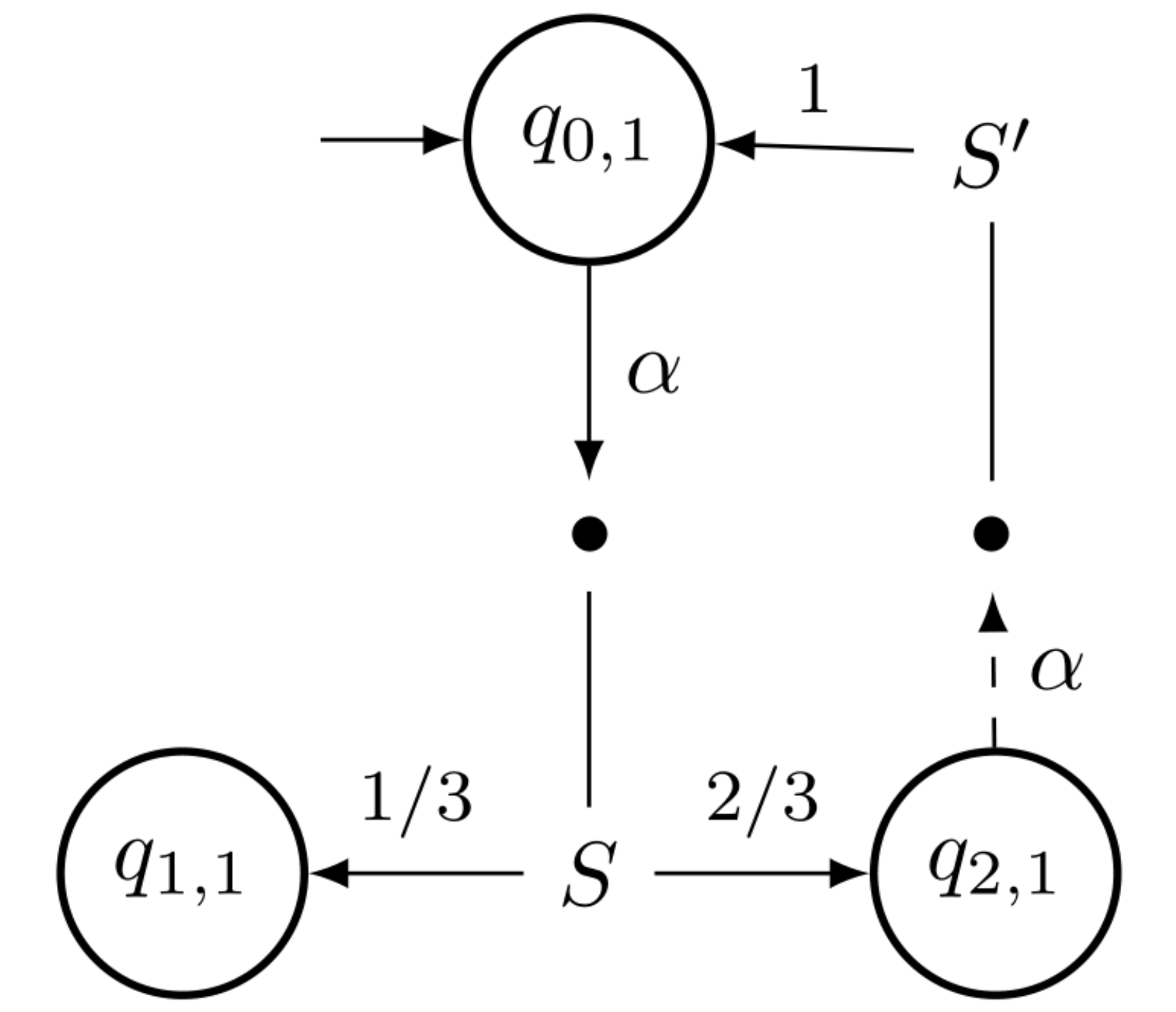}
     \end{center}
\end{minipage} \hfill
 \begin{minipage}[c]{.48\linewidth}
   \begin{center}
     \includegraphics[scale=0.3]{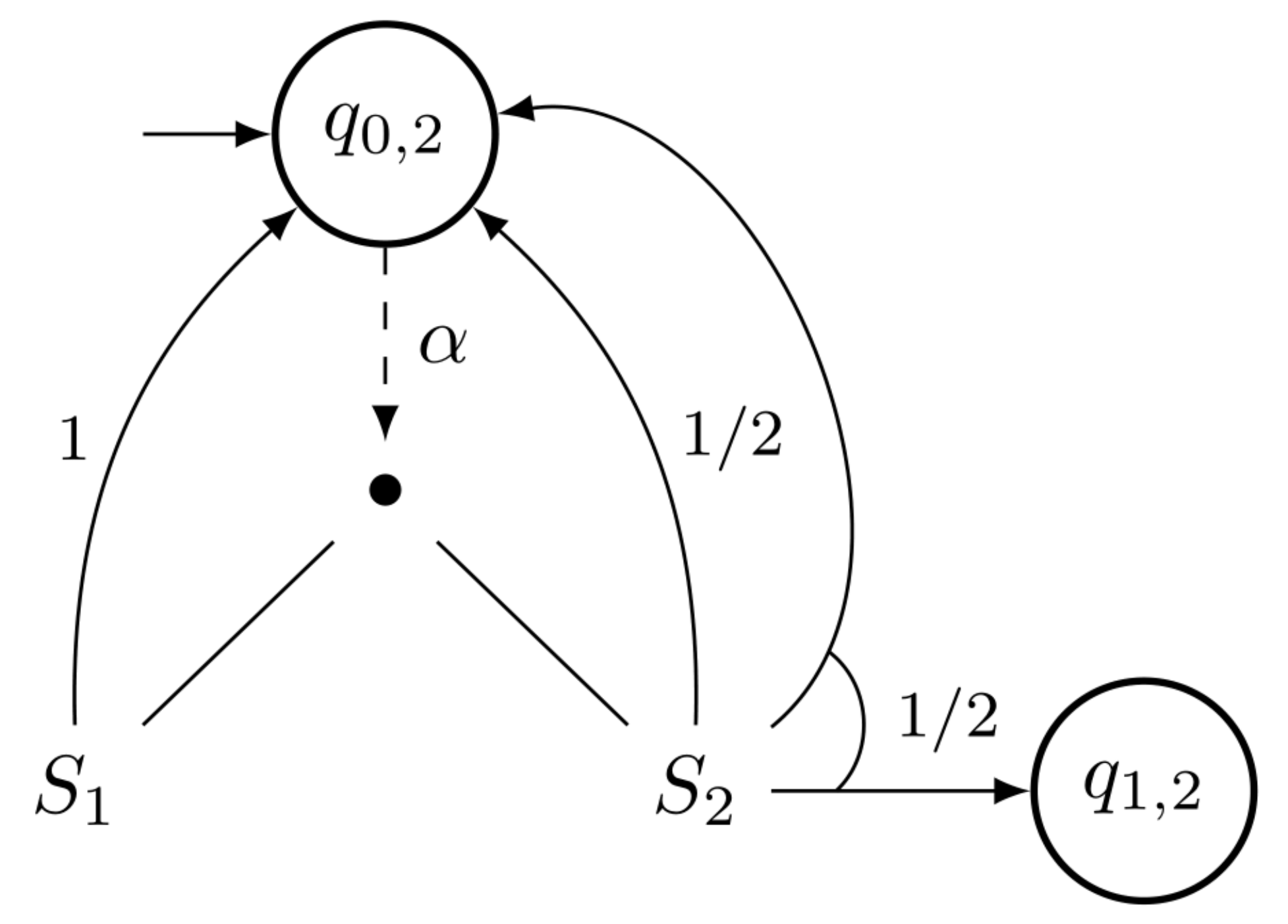}
 \end{center}
 \end{minipage}
 \caption{\sf Example of refinement, see Examples~\ref{rleuglkkkjh} and~\ref{kltrughltuih}.}\label{fig:ex-ref}
\end{figure}
\begin{example}\rm \label{kltrughltuih}
  Figure~\ref{fig:ex-ref} shows an example of refinement. The \mmi\;
  on the left is a refinement of the one on the right. Observe in
  particular that the \mmi\; on the left still encompasses
  probabilistic aspects but no longer has nondeterministic select for
  the next state. This is allowed by the lifting operation on mixed
  systems as already seen in Example~\ref{ex:lift}.\eproof
\end{example}

\begin{lemma} \label{eorgfuirhpu}
The modal refinement on \mmi{s} is a preorder.
\end{lemma}
\begin{proof}
	 See Appendix~\ref{guioenhuio}.\eproof
\end{proof}
\begin{theorem} \label {oer8w5t7ho0u}
  For $\contract_i,i=1,2$ two
  \mmi{s}, if $\contract_1\refines\contract_2$ then every model of
  $\contract_1$ is also a model of $\contract_2$.
\end{theorem}
\begin{proof}
	 See Appendix~\ref{eluihuihipgh}.\eproof
\end{proof}
Despite \mmi{s} are taken deterministic in Definition~\ref{def-empi},
modal refinement is correct but not fully abstract as for Modal
Automata~\cite{LarsenNW07CONCUR}: the following counterexample shows
that Theorem~\ref{oer8w5t7ho0u} cannot be strengthened to an
if-and-only-if statement. The reason for this is the nondeterminism
that sits in the mixed systems themselves.
\begin{cexample}\rm \label{kerughlutkhk}
  Consider the two ``purely non-probabilistic'' \mmi{s} over
  $\alphabet=\{a\}$ depicted in Figure~\ref{fig:incompl-ref}. They are
  purely non-probabilistic as any associated random follows a Dirac
  probability.  $\contract_1$ has only models that can perform at most
  two consecutive $\action$-actions. Any such implementation is also
  an implementation of $\contract_2$. However, it is not true that
  $\contract_1\refines\contract_2$ in the sense of modal refinement.\eproof
\begin{figure}[ht]
  \begin{minipage}[c]{.48\linewidth}
   \begin{center}
     \includegraphics[scale=0.3]{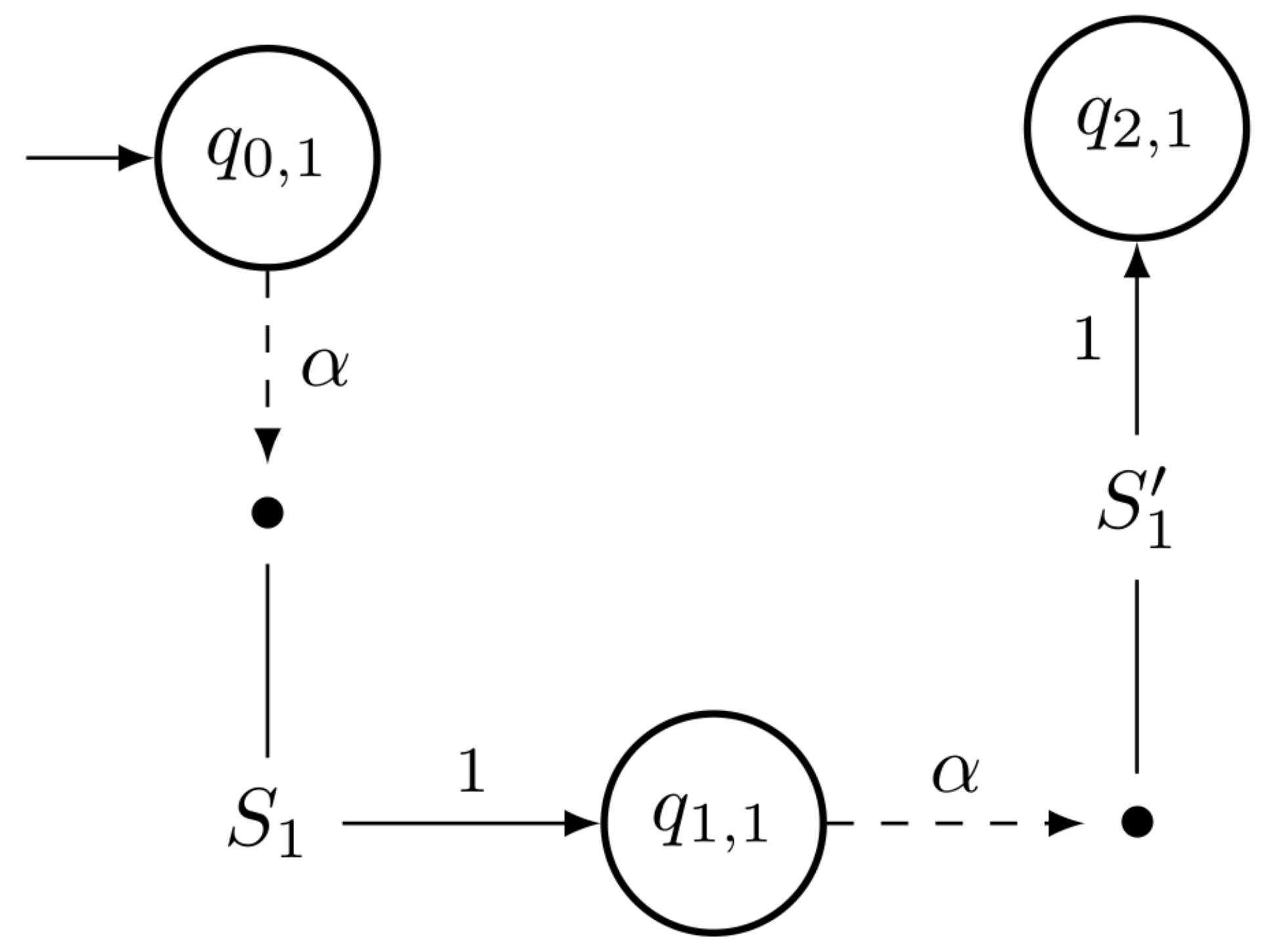}
     \end{center}
\end{minipage} \hfill
 \begin{minipage}[c]{.48\linewidth}
   \begin{center}
     \includegraphics[scale=0.3]{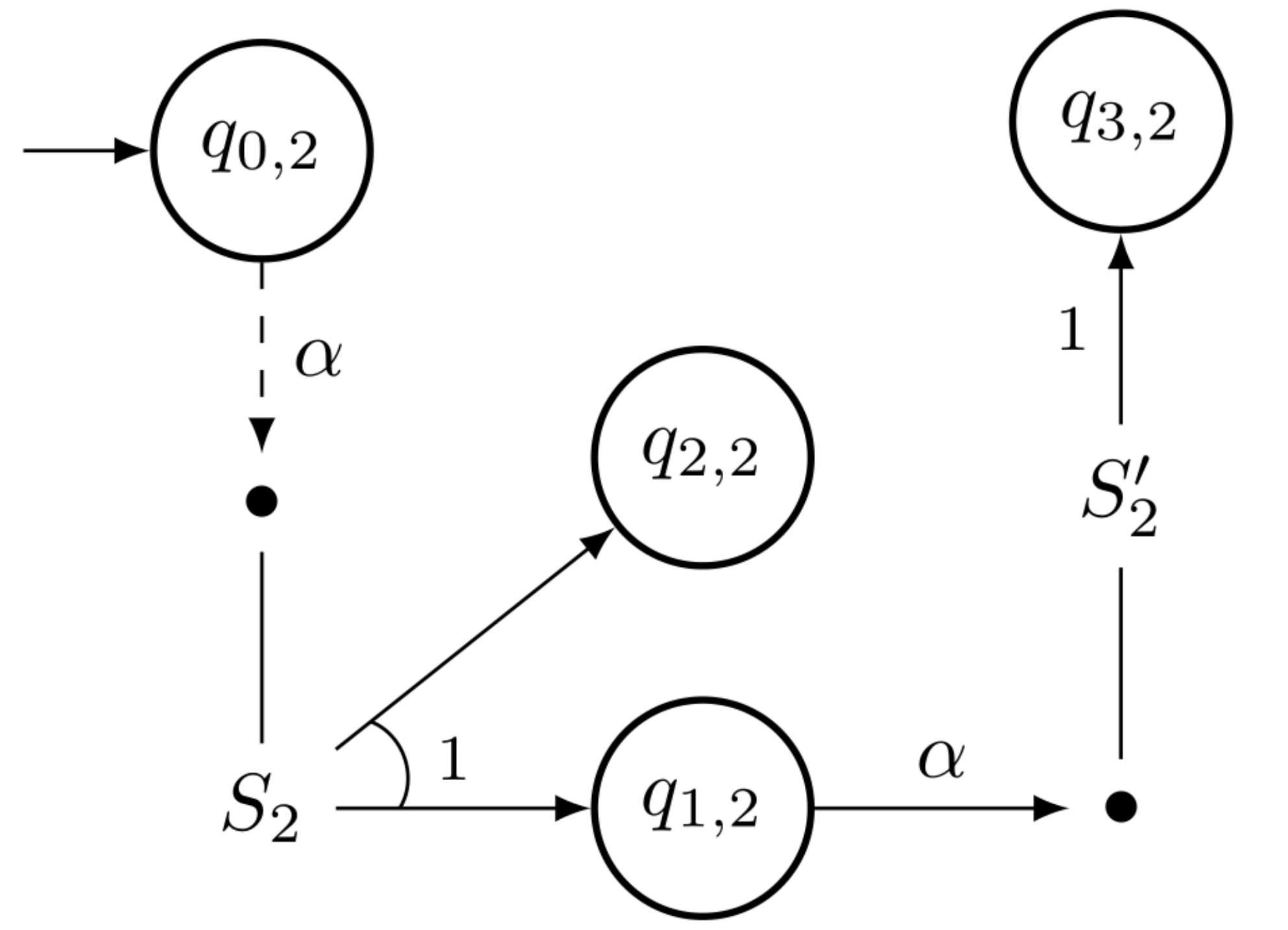}
 \end{center}
 \end{minipage}
 \caption{\sf Counterexample~\ref{kerughlutkhk} showing that modal refinement is not fully abstract}\label{fig:incompl-ref}
\end{figure}
\end{cexample}
\paragraph{Conjunction.}
Consider $\contract_i,i=1,2$ two {\mmi{s}} over $\alphabet$ with
respective sets of variables $X_1$ and $X_2$ and state spaces $Q_1$
and $Q_2$.
\beq
\Systems_1\times\Systems_2 &=& \left\{
\system_1\times\system_2
\mid{\system_i}\in\Systems_i
\right\}
\label {ortguhoerop68}
\eeq
where $\system_1\times\system_2$ is defined in Definition~\ref{lrtguiodtrhleguip}.

We are now able to define the conjunction of two \mmi{s}.
\begin{definition}[conjunction]\label{def-conjunction} \it
  Let $\contract_i,i=1,2$ be two {\mmi{s}} over $\alphabet$, their
  \emph{pre-conjunction} $\contract_1\underline{\wedge}\contract_2$
  has alphabet $\alphabet$, set of variables $X_1{\cup}X_2$, initial
  state $(q_{1,0},q_{2,0})$, and its \may\ and \must\ transition
  relations are the minimal relations satisfying the following rules:
   \beqq\bea{lclclcl}
\mbox{\emph{[ConjMay]}}&:&
\transmayindex{q_1}{\action}{\Systems^\probamay_1}{1}
&\emph{and}&
\transmayindex{q_2}{\action}{\Systems^\probamay_2}{2}
 &\Ra&
\transmayindex{(q_1,q_2)}{\action}{\Systems^\probamay_1{\times}\Systems^\probamay_2}{}
\\ \vspace*{-3mm} \\
\mbox{\emph{[ConjMust0]}}&:&
\transmustindex{q_1}{\action}{\Systems^\probamust_1}{1}
&\emph{and}&
\transmustindex{q_2}{\action}{\Systems^\probamust_2}{2}
 &\Ra&
\transmustindex{(q_1,q_2)}{\action}{\Systems^\probamust_1{\times}\Systems^\probamust_2}{}
\\ \vspace*{-3mm} \\
\mbox{\emph{[ConjMust1]}}&:&
\transmustindex{q_1}{\action}{\Systems^\probamust_1}{1}
&\emph{and}&
\ntransmustindex{q_2}{\action}{}{2}
 &\Ra&
\transmustindex{(q_1,q_2)}{\action}{\Systems^\probamust_1{\times}\Systems(X_2)}{}
\\ \vspace*{-3mm} \\
\mbox{\emph{[ConjMust2]}}&:&
\ntransmustindex{q_1}{\action}{}{1}
&\emph{and}&
\transmustindex{q_2}{\action}{\Systems^\probamust_2}{2}
 &\Ra&
\transmustindex{(q_1,q_2)}{\action}{\Systems(X_1){\times}\Systems^\probamust_2}{}
\eea
\eeqq
Pruning for consistency the pre-conjunction $\contract_1\underline{\wedge}\contract_2$ yields the \emph{conjunction} $\contract_1{\wedge}\contract_2$.
\end{definition}
Inconsistency may result from the rules [ConjMust1] and [ConjMust2].

\begin{theorem} \label{kuygpgiohtoi}
  \it   For any \mmi\ $\contract_1$ and $\contract_2$, any model of $\contract_1 \wedge \contract_2$ is also a model of
$\contract_1$ and $\contract_2$.
\end{theorem}
\begin{proof}
	 See Appendix~\ref{erouigheguio}.\eproof
\end{proof}
\paragraph{Parallel composition.}
Quite often in the literature, an issue of \emph{compatibility} arises along with the parallel composition of interfaces~\cite{DBLP:conf/emsoft/AlfaroH01,Raclet2011a}. As clarified in~\cite{Raclet2011a}, the issue of compatibility is due to the different roles played by the component and its environment in dealing with inputs and outputs. As we do not distinguish inputs and outputs here, compatibility is not an issue for us.

\begin{definition}[composition] \label {g07hortgig}
  Let $\contract_i,i=1,2$ be two {\mmi{s}} over $\alphabet$, their
  \emph{composition} $\contract_1\cpara\contract_2$ has alphabet $\alphabet$, set of variables $X_1{\cup}X_2$, and initial state $(q_{1,0},q_{2,0})$. Its transition relations are the minimal relations satisfying the following rules:
  \beqq
\transmayindex{q_1}{\action}{\Systems^\probamay_1}{1}
\emph{ and }
\transmayindex{q_2}{\action}{\Systems^\probamay_2}{2}
&\ \Ra \ &
\transmayindex{(q_1,q_2)}{\action}{\Systems^\probamay_1{\times}\Systems^\probamay_2}{}
\\
\transmustindex{q_1}{\action}{\Systems^\probamust_1}{1}
\emph{ and }
\transmustindex{q_2}{\action}{\Systems^\probamust_2}{2}
& \ \Ra \ &
\transmustindex{(q_1,q_2)}{\action}{\Systems^\probamust_1{\times}\Systems^\probamust_2}{}
\eeqq
\end{definition}
Parallel Composition does not raise any issue of consistency.
\begin{theorem} \label {rtpery9uy95e8}
The parallel composition $\cpara$ satisfies the following properties:
\begin{compactenum}
\item \label{jiytdfiout}
$\cpara$ is commutative and associative.
	\item \label{wieytdfiwet} For $\contract_1$ and $\contract_2$
          two \mmi{s}, we have:
          \vspace{-0.1cm}
\beq
\forall M_i, i=1,2: M_i\models\contract_i
&\Ra& M_1{\times}M_2\models\contract_1\cpara\contract_2
\label {lrtuhlup}
\\
\contract_1 \refines
  \contract_2
  &\Ra& \forall\contract: \contract \cpara \contract_1 \refines
  \contract \cpara \contract_2
  \label {rltiugyhpe}
\eeq
\end{compactenum}
\end{theorem}
\begin{proof}
	 See Appendix~\ref{epiougtrhpiough}.\eproof
\end{proof}
Last, let us mention that no quotient exists for Mixed
Interfaces. This is inherently due to the nondeterminism involved in
Mixed Systems. Probabilistic specification models already suffer from
the same limitation.

\section{Link to Constraint Markov Chains}
\label{rtuiothoguih}
Constraint Markov Chains have been proposed
in~\cite{DBLP:journals/tcs/CaillaudDLLPW11} as a specification
formalism with Markov Chains as models. Let us first recall their
basic definitions.
%

Let $A,B$ be sets of propositions with $A\subseteq{B}$. The
restriction of $W\subseteq{B}$ to $A$ is given by
$\restrict{A}{W}=W\cap{A}$. If $T\subseteq{2^B}$, then
$\restrict{A}{T}=\{\restrict{A}{W}\mid W\in{T}\}$. Let $\cP(Q)$
denote the set of all probabilities over the set $Q$. For $R$ and $Q$
two at most denumerable state spaces, a \emph{transition probability}
$\Delta$, from $R$ to $Q$, is a map \mbox{$\Delta:R{\times}Q\ra[0,1]$}
such that, for every $r\in{R}$, $\Delta(r,.)$ is a probability over
$Q$. If $\proba_R$ is a probability distribution over $R$, then
$\proba_R\Delta$ denotes the probability distribution over $Q$ defined
by:
\beq \proba_R\Delta(q)&=&\sum_{r\in{R}}\proba_R(r)\Delta(r,q)\,.
\label{wcuetrcde}
\eeq
A \emph{transition sub-probability} $\Delta$ from $R$ to $Q$ is a
map $\Delta:R{\times}Q\ra[0,1]$ such that, for every
$(r,q)\in{R}{\times}Q$, $\Delta(r,q)\geq{0}$ and, for every $r\in{R}$,
$\sum_{q\in{Q}}\Delta(r,q)\leq{1}$.


\begin{definition}\label{liugtil}
  A \emph{Markov Chain (\mcs)} is a tuple
$\mc=(R,r_0,\Proba,A,v)$,
where $R$ is a set of states containing the initial state $r_0$, $A$
is a set of atomic propositions, $v:R\ra{2^A}$ is a state valuation,
and $\Proba:{R}{\times}R{\ra}[0,1]$ is a transition probability.
\end{definition}

\begin{definition}\label{ltgouihuguih}
A \emph{Constraint Markov Chain (\cmcs)} is a tuple
\[
\cmc=(Q,q_0,\varphi,A,V)\,,
\]
where $Q$ is a set of states containing the initial state $q_0$, $A$
is a set of atomic propositions, $V:Q{\ra}2^{2^A}$ is a set of
admissible state valuations, and $\varphi:{Q}{\ra}2^{\cP(Q)}$ is a
\emph{constraint function,} mapping states to sets of probability
distributions over states.
\end{definition}
In practice, constraint functions will be only partially specified, in
that a function mapping $Q$ to $[0,1]^Q$ will be implicitly
complemented by the additional constraints to make the target being a
probability. This consideration is only practical and does not need to
be taken into account for our subsequent development.
Whenever needed to avoid confusion, we will denote by $A_\mc$ and
$v_\mc$, and $A_\cmc$ and $V_\cmc$, the elements $A$ and $V$ of $\mcs$
$\mc$ and $\cmcs$ $\cmc$.

\begin{definition}[satisfaction]
  \label{ywtfqwkfqkuqf} Let $\mc$ and $\cmc$ be respectively an $\mcs$
  and a $\cmcs$ such that $A_\cmc\subseteq{A_\mc}$.  A
  \emph{satisfaction relation} between $\mc$ and $\cmc$ is a relation
  $\simu\subseteq{R}{\times}{Q}$ such that, whenever $r\,\simu\;{q}$:
\begin{compactenum}
\item \label{dtyuywtdf} $\restrict{A_\cmc}{v_\mc(r)}\in V_\cmc(q)$;
\item \label{detfwiyt} there exists a {transition sub-probability}
  $\Delta$, from $R$ to $Q$, such that:
\begin{compactenum}
\item \label{uetydyut} for all $r'\in{R}$ such that $\Proba(r,r')>0$,
  $\Delta(r',q)$ is a transition probability from $R$ to $Q$, and;
\item $\Proba(r,.)\Delta\in\varphi(q)$, and;
\item \label{ioucdgfouy} if $\Delta(r',q')\neq{0}$, then
  $r'\simu\;{q'}$ holds.
\end{compactenum}
\end{compactenum}
$\mc$ satisfies $\cmc$ if and only if there exists a satisfaction
relation between $\mc$ and $\cmc$ that contains the two initial
states.
\end{definition}


\begin{definition}[weak refinement]
  \label{dweiytdfutdrfyf}
  Let $\cmc_1$ and $\cmc_2$ be two \cmcs\ such that
  $A_2{\subseteq}{A_1}$. The relation
  $\simu\subseteq{Q_1}{\times}{Q_2}$ is a \emph{weak refinement} iff,
  whenever $q_1\,\simu\;q_2$:
\begin{compactenum}
\item \label{ieufgierufyg}
  $\restrict{A_2}{V_1(q_1)}\subseteq{V_2(q_2)}$;
\item \label{idfgkukeyufgr} for any probability distribution
  $\proba_1\in\varphi_1(q_1)$, there exists a transition
  sub-probability $\Delta$, from $Q_1$ to $Q_2$, such that:
\begin{compactenum}
\item \label{efgeorufygeu} for all $q_1$ such that $\proba_1(q_1)>0$,
  $\Delta(q_1,.)$ is a probability over $Q_2$;
\item \label{efyugoyu} $\proba_1\Delta\in\varphi_2(q_2)$;
\item \label{geyugfwkugfwyt} if $\Delta(q'_1,q'_2)>0$, then
  $q'_1\,\simu\;q'_2$ holds.
\end{compactenum}
\end{compactenum}
We say that $\cmc_1$ \emph{weakly refines} $\cmc_2$, written
$\cmc_1{\preceq}\cmc_2$, if $q_{0,1}\;\simu\ q_{0.2}$.
\end{definition}

We now show that \mmi{s} subsume \cmcs. First, we define the embedding
of \mcs\ in \mmdp. Given $\mc=(R,r_0,\Proba,A,v)$ a Markov Chain, we
associate the \mmdp\
\mbox{$M_\mc=(\alphabet,X,r_0,\ra)$,}
where:
\begin{compactitem}
\item $\alphabet=\{\action\}$ (no need to mention the only
  action labeling transitions);
\item $X=\{\xi,v\}$ collects a variable $\xi$ with domain $R$,
  and the variable $v$;
\item $r_0\in{R}$ is the initial condition for $\xi$; no initial
  condition is given for $v$;
\item the transition relation is $\trans{r}{}{\system}{}$, where the
  mixed system $S=((\Omega,\proba),X,\cons)$ is such that:
  \vspace{-0.1cm}
  \beq
	\Omega{=}R \;;\; \proba{=}\Proba(r,.) \;\mbox{ and }\;
\cons = \left\{
(r,r,v(r))
\mid
r{\in}{R}
\right\} \subseteq \Omega{\times}({R}{\times}2^A)
\label{ewyjtdfwkdtf}
\eeq
\end{compactitem}

\begin{lemma}
  \label{elriueoiufhyi}
  Let $\mc$ be an $\mcs$. Then, $\mc$ and $M_\mc$ possess identical
  semantics.
\end{lemma}
\begin{proof}
	 See Appendix~\ref{eroiuhpiaeug}.\eproof
\end{proof}
Consider now the embedding of \cmcs\ in \mmi{s}. For any \cmcs\
$\cmc=(Q,q_0,\varphi,A,V)$, we associate a \mmi\
$\contract_\cmc=({\alphabet},X,q_0,\ramust,\ramay)$,
where:
\begin{compactitem}
\item $\alphabet=\{\action\}$ (no need to mention the only
  action labeling transitions);
\item $X=\{\xi,v\}$ collects a variable $\xi$ with domain $Q$,
  and a variable $v$ with domain $2^A$;
\item $q_0\in{Q}$ is the initial condition for $\xi$; no initial
  condition is given for $V$;
\item the \emph{must} transition relation $\ramust$ is empty;
\item the \emph{may} transition relation is
  $\transmay{q}{}{\Systems}{}$, where $\Systems$ is the set of mixed
  systems of the form $S_v=((\Omega,\proba),X,\cons_v)$, where $v(q)$
  ranges over ${V(q)}$ and:
  \beq 
	\Omega{=}Q \;;\; \proba{\in}\varphi(q)
  \;\mbox{ and }\;
\cons_v = \left\{
(q,q,v(q))
\mid
q{\in}{Q}
\right\} ~\subseteq~ \Omega{\times}({R}{\times}2^A)
\label{uytdfiytfiyt}
\eeq
\end{compactitem}
Whenever needed, we will use subscripts to relate items of $\mc$ and
$\cmc$ to their respective host entities.

\begin{theorem}
  \label{leprihiuhi}
  Let $\cmc$ be a $\cmcs$. Then, $\cmc$ and $\contract_\cmc$ possess
  identical semantics.
\end{theorem}
The previous Theorem decomposes into the two following lemmas.

\begin{lemma} \label{eruiehrpiufhio} Let $\mc$ and $\cmc$ be
  respectively an $\mcs$ and a $\cmcs$ such that
  $A_\cmc{\subseteq}{A_\mc}$. Then, $\mc$ satisfies $\cmc$ iff $M_\mc$
  is a model of $\contract_\cmc$.
\end{lemma}
\begin{proof}
	See Appendix~\ref{eoruifygouiy}.\eproof
\end{proof}

\begin{lemma}
  \label{elriuhepiu} Let $\cmc_1$ and $\cmc_2$ be two \cmcs{s} such
  that $A_2\subseteq{A_1}$. Then, $\cmc_1$ weakly refines $\cmc_2$ iff
  $\contract_{\cmc_1}$ refines $\contract_{\cmc_2}$.
\end{lemma}
\begin{proof}
	See Appendix~\ref{erlivgfuebhiu}.\eproof
\end{proof}

\section{Conclusion}
We have proposed the first interface theory that allows to mix probabilities and nondeterminism. Our component model is that of Mixed Markov Decision Processes (\mmdp) which subsume Probabilistic Automata. Our specification formalism is that of Mixed Interfaces. It offers a complete algebra for interfaces, namely: satisfaction, refinement, conjunction, and parallel composition. No quotient exists for Mixed Interfaces. This is inherently due to the nondeterminism involved in Mixed Systems. We presented our framework for the case of a fixed alphabet of actions. Following~\cite{Raclet2011a}, alphabet extension techniques allow to handle the general case, this will be reported in the extended version of this work.

Mixed Interfaces extend and clarify the satisfaction and refinement relations defined for Constraint Markov Chains. The same holds for Abstract
Probabilistic Automata (APA)~\cite{DBLP:conf/vmcai/DelahayeKLLPSW11}. CMC and APA differ from Mixed Interfaces regarding the parallel composition, however. The parallel composition for Mixed Interfaces is general (system variables can be shared), whereas the one for CMC or APA requires that the specifications for composition have disjoint sets of atomic propositions. Also, a subclass of Mixed Interfaces can be defined that tightly emulates the networks of \emph{Price Timed Automata} (\pta) equipped with their stochastic semantics~\cite{DBLP:journals/corr/abs-1106-3961}; a complete emulation, however, requires the consideration of some non-compositional priority policy for closed systems in this subclass. Due to lack of space, these additional results were not presented here.

This paper sets the theoretical foundations of formalisms that we plan
to apply to safety and vulnerability analysis as ongoing works. To
make it effective and amenable of tool development, one step further
is needed, namely a finitary syntax for specifying and manipulating
sets of Mixed Systems.



\bibliographystyle{plain}
\bibliography{proba}

\begin{thebibliography}{10}

\bibitem{AHLNW-08}
Adam Antonik, Michael Huth, Kim~G. Larsen, Ulrik Nyman, and Andrzej Wasowski.
\newblock 20 years of modal and mixed specifications.
\newblock {\em Bulletin of European Association of Theoretical Computer
  Science}, 1(94), 2008.

\bibitem{BaierK00}
Christel Baier and Marta~Z. Kwiatkowska.
\newblock Domain equations for probabilistic processes.
\newblock {\em Mathematical Structures in Computer Science}, 10(6):665--717,
  2000.

\bibitem{BauerDHLLNW12}
Sebastian~S. Bauer, Alexandre David, Rolf Hennicker, Kim~Guldstrand Larsen,
  Axel Legay, Ulrik Nyman, and Andrzej Wasowski.
\newblock Moving from specifications to contracts in component-based design.
\newblock In {\em Proc. of the 15th International Conference on Fundamental
  Approaches to Software Engineering (FASE'12)}, volume 7212 of {\em Lecture
  Notes in Computer Science}, pages 43--58. Springer, 2012.

\bibitem{DBLP:conf/fmco/BenvenisteCFMPS07}
Albert Benveniste, Beno\^{\i}t Caillaud, Alberto Ferrari, Leonardo Mangeruca,
  Roberto Passerone, and Christos Sofronis.
\newblock Multiple viewpoint contract-based specification and design.
\newblock In {\em Proc. of the 6th International Symposium on Formal Methods
  for Components and Objects (FMCO'06)}, volume 5382 of {\em Lecture Notes in
  Computer Science}, pages 200--225. Springer, 2007.

\bibitem{DBLP:journals/fteda/BenvenisteCNPRR18}
Albert Benveniste, Beno{\^{\i}}t Caillaud, Dejan Nickovic, Roberto Passerone,
  Jean{-}Baptiste Raclet, Philipp Reinkemeier, Alberto~L.
  Sangiovanni{-}Vincentelli, Werner Damm, Thomas~A. Henzinger, and Kim~G.
  Larsen.
\newblock Contracts for system design.
\newblock {\em Foundations and Trends in Electronic Design Automation},
  12(2-3):124--400, 2018.

\bibitem{BenvenisteLFG95}
Albert Benveniste, Bernard~C. Levy, Eric Fabre, and Paul~Le Guernic.
\newblock A calculus of stochastic systems for the specification, simulation,
  and hidden state estimation of mixed stochastic/nonstochastic systems.
\newblock {\em Theor. Comput. Sci.}, 152(2):171--217, 1995.

\bibitem{DBLP:conf/qest/CaillaudDLLPW10}
Beno\^{\i}t Caillaud, Beno\^{\i}t Delahaye, Kim~G. Larsen, Axel Legay,
  Mikkel~L. Pedersen, and Andrzej Wasowski.
\newblock Compositional design methodology with constraint markov chains.
\newblock In {\em Proc. of the 7th International Conference on the Quantitative
  Evaluation of Systems (QEST'07)}, pages 123--132. IEEE Computer Society,
  2010.

\bibitem{DBLP:journals/tcs/CaillaudDLLPW11}
Beno\^{\i}t Caillaud, Beno\^{\i}t Delahaye, Kim~G. Larsen, Axel Legay,
  Mikkel~L. Pedersen, and Andrzej Wasowski.
\newblock Constraint markov chains.
\newblock {\em Theor. Comput. Sci.}, 412(34):4373--4404, 2011.

\bibitem{DBLP:journals/corr/abs-1106-3961}
Alexandre David, Kim~G. Larsen, Axel Legay, Marius Mikucionis, Danny~B{\o}gsted
  Poulsen, Jonas van Vliet, and Zheng Wang.
\newblock Stochastic semantics and statistical model checking for networks of
  priced timed automata.
\newblock {\em CoRR}, abs/1106.3961, 2011.

\bibitem{DBLP:conf/emsoft/AlfaroH01}
Luca de~Alfaro and Thomas~A. Henzinger.
\newblock Interface theories for component-based design.
\newblock In {\em Proc. of the 1st International Workshop on Embedded Software
  (EMSOFT'01)}, volume 2211 of {\em Lecture Notes in Computer Science}, pages
  148--165. Springer, 2001.

\bibitem{DBLP:conf/vmcai/DelahayeKLLPSW11}
Beno\^{\i}t Delahaye, Joost-Pieter Katoen, Kim~G. Larsen, Axel Legay, Mikkel~L.
  Pedersen, Falak Sher, and Andrzej Wasowski.
\newblock Abstract probabilistic automata.
\newblock In {\em Proc. of the 12th International Conference on Verification,
  Model Checking, and Abstract Interpretation (VMCAI'11)}, volume 6538 of {\em
  Lecture Notes in Computer Science}, pages 324--339. Springer, 2011.

\bibitem{DBLP:conf/acsd/DelahayeKLLPSW11}
Beno\^{\i}t Delahaye, Joost-Pieter Katoen, Kim~G. Larsen, Axel Legay, Mikkel~L.
  Pedersen, Falak Sher, and Andrzej Wasowski.
\newblock New results on abstract probabilistic automata.
\newblock In {\em Proc. of the 11th International Conference on Application of
  Concurrency to System Design (ACSD'11)}, pages 118--127. IEEE, 2011.

\bibitem{DBLP:conf/qest/DelahayeLLPW11}
Beno\^{\i}t Delahaye, Kim~G. Larsen, Axel Legay, Mikkel~L. Pedersen, and
  Andrzej Wasowski.
\newblock Apac: A tool for reasoning about abstract probabilistic automata.
\newblock In {\em Proc. of the 8th International Conference on Quantitative
  Evaluation of Systems (QEST'11)}, pages 151--152. IEEE Computer Society,
  2011.

\bibitem{Derman70}
Cyrus Derman.
\newblock {\em Finite state Markovian decision processes}.
\newblock Academic Press, 1970.

\bibitem{DBLP:conf/lics/JonssonL91}
Bengt Jonsson and Kim~Guldstrand Larsen.
\newblock Specification and refinement of probabilistic processes.
\newblock In {\em Proc. of the 6th Annual Symposium on Logic in Computer
  Science (LICS'91)}, pages 266--277. IEEE Computer Society, 1991.

\bibitem{Larsen89}
Kim~Guldstrand Larsen.
\newblock Modal specifications.
\newblock In {\em Automatic Verification Methods for Finite State Systems},
  volume 407 of {\em Lecture Notes in Computer Science}, pages 232--246.
  Springer, 1989.

\bibitem{LarsenNW07CONCUR}
Kim~Guldstrand Larsen, Ulrik Nyman, and Andrzej Wasowski.
\newblock On modal refinement and consistency.
\newblock In {\em Proc. of the 18th Inter. Conf. on Concurrency Theory
  (CONCUR'07)}, pages 105--119. Springer, 2007.

\bibitem{KT88}
Kim~Guldstrand Larsen and Bent Thomsen.
\newblock A modal process logic.
\newblock In {\em Proc. of the 3rd Annual Symposium on Logic in Computer
  Science (LICS'88)}, pages 203--210. IEEE, 1988.

\bibitem{DBLP:conf/concur/LynchSV03}
Nancy~A. Lynch, Roberto Segala, and Frits~W. Vaandrager.
\newblock Compositionality for probabilistic automata.
\newblock In {\em Proc. of the 14th International Conference on Concurreny
  Theory (CONCUR'03)}, volume 2761 of {\em Lecture Notes in Computer Science},
  pages 204--222. Springer, 2003.

\bibitem{Puterman}
M.~L. Puterman.
\newblock {\em Markov {D}ecision {P}rocesses}.
\newblock J. Wiley and Sons, 1994.

\bibitem{Raclet2011a}
Jean-Baptiste Raclet, Albert Benveniste, Beno\^{\i}t Caillaud, Axel Legay, and
  Roberto Passerone.
\newblock A modal interface theory for component-based design.
\newblock {\em Fundamenta Informaticae}, 107:1--32, 2011.

\bibitem{Seg06-CONCUR}
Roberto Segala.
\newblock Probability and nondeterminism in operational models of concurrency.
\newblock In {\em Proc. of the 17th International Conference on Concurrency
  Theory (CONCUR'06)}, volume 4137 of {\em Lecture Notes in Computer Science},
  pages 64--78. Springer, 2006.

\bibitem{SegalaL94}
Roberto Segala and Nancy~A. Lynch.
\newblock Probabilistic simulations for probabilistic processes.
\newblock In {\em Proc. of the 5th International Conference on Concurrency
  Theory (CONCUR'94)}, volume 836 of {\em Lecture Notes in Computer Science},
  pages 481--496. Springer, 1994.

\end{thebibliography}

\newpage


 \appendix

\section{Proofs regarding Mixed Systems}
\subsection{Proof of Lemma~\ref{wjdetyfuuy}}
\label{riuygfttyiohjih}
\begin{proof} It is enough to prove the result for compressed systems.
  For $i=1,2$, let $\system_i\equiv\system'_i$ and let $\varphi_i$ be
  the bijections defining the two equivalences. With reference to
  (\ref{usdqwcdur}), we define \beqq \varphi(\omega,q_1\join{q_2}) &=&
  \left( (\omega'_1,\omega'_2),q'_1\join{q'_2} \right) \mbox{ where }
  (\omega'_i,q'_i) =\varphi_i(\omega_i,q_i), i=1,2 \eeqq and we have
  to verify that $\varphi$ defines the desired equivalence between
  $\system\eqdef\system_1\mpara\system_2$ and
  $\system'\eqdef\system'_1\mpara\system'_2$. Using the expression
  (\ref{usdqwcdur}) for $\cons$ and the fact that
  $\proba=\proba_1\otimes\proba_2$, we get
	 \[\bea{rl}
	 \cons_\proba=&\{
(\omega,q_1\join{q_2})
\mid
q_1\compat{q_2} \,\wedge\, 
\omega_1\cons_1{q_1} \,\wedge\, \proba_1(\omega_1)>0 \,\wedge\, 
\omega_2\cons_2{q_2} \,\wedge\,  \proba_2(\omega_2)>0 
\} \\
=& \{
(\omega,q_1\join{q_2})
\mid
q_1\compat{q_2} \,\wedge\, 
(\omega_1,q_1)\in\cons_{1\proba}  \,\wedge\, 
(\omega_2,q_2)\in\cons_{2\proba}
\}
\eea
	 \]
	 Thus, for every $(\omega,q_1\join{q_2})\in\cons_\proba$, we have
$q'_1=q_1\compat{q_2}=q'_2  \mbox{ and }
(\omega'_i,q'_i)\in\cons_{i\proba}, i=1,2$,
whence $(\omega',q')\in\cons'_{\proba}$ and $\varphi$ is a bijection. Since $\proba'=\proba'_1\otimes\proba'_2$ we get $\proba'(\omega')=\proba(\omega)$, which finishes the proof.\eproof
\end{proof}

\subsection{Proof of Lemma~\ref{egfuioehrpo}}
\label{elrgfuilyu}
\begin{proof}
The result is immediate if both $\system_1$ and $\system'_1$ are compressed, see Definition~\ref{lighlalegfr}. It is thus sufficient to prove the lemma for the following two particular cases: $\system_1$ compresses to $\system'_1$, and the converse. 

Consider first the case: $\system_1$ compresses to $\system'_1$. Let $w(\omega_1,\omega_2)$ be the weighting function associated to the lifting $\system_1\NMPlift{\simu}\system_2$, and let $\proba'_1(\omega'_1)=\sum_{\omega_1\in\omega'_1}\proba_1(\omega_1)$ be the relation between $\proba'_1$ and $\proba_1$ in the compression of $\system_1$ to $\system'_1$. Then $w'(\omega'_1,\omega_2)=\sum_{\omega_1\in\omega'_1}w(\omega_1,\omega_2)$ defines the weighting function associated to the lifting $\system'_1\NMPlift{\simu}\system_2$. The other properties required to deduce $\system'_1\NMPlift{\simu}\system_2$ are immediate to prove.

Now, consider the alternative case: $\system'_1$ compresses to $\system_1$, with relation 
\beq\bea{c}
\proba_1(\omega_1)=\sum_{\omega'_1\in\omega_1}\proba'_1(\omega'_1)
\eea
\label{ofoihiouhgliuh}
\eeq
 between $\proba'_1$ and $\proba_1$, where $\omega'_1\in\omega_1$ means that $\omega_1$ is the equivalence class of $\omega'_1$ with respect to relation $\sim$ defined in (\ref{eoguheogihio}) when compressing $\system'_1$. This case is more involved since the construction of the weighting function $w'(\omega'_1,\omega_2)$ is nontrivial. We need $w'(\omega'_1,\omega_2)$ to satisfy the following relations:
\begin{equation}\bea{rl}
\forall \omega'_1:&
\proba'_1(\omega'_1)=\sum_{\omega_2}w'(\omega'_1,\omega_2)
\\
\forall \omega_2:&
\proba_2(\omega_2)=\sum_{\omega'_1}w'(\omega'_1,\omega_2)
\\
[2mm]
\forall(\omega'_1,\omega_2;q_1):&
\left[\bea{c}
w'(\omega'_1,\omega_2)>0 \\ \omega'_1\,\cons'_1\,{q_1}\eea\right] \Ra \exists q_2:\left[\bea{c}\omega_2\,\cons_2\,{q_2} \\ q_1\,\simu\,{q_2}\eea\right]
\\ 
[-3mm]
\wemph{.}
\eea
\label{eruithpeuhpoh}
\end{equation}
Focus first on the first two lines of (\ref{eruithpeuhpoh}). We claim that to find a solution $w' $ to the first two lines of (\ref{eruithpeuhpoh}), it is enough to find a solution to the following system of equations where the unknowns are the values $w'(\omega'_1,\omega_2)$:
\beq
\bea{rrcl}
\forall\omega_1,\omega_2:&
\sum_{\omega'_1\in\omega_1}w'(\omega'_1,\omega_2)&=&w(\omega_1,\omega_2)
\\
[1mm]
\forall \omega'_1:&
\sum_{\omega_2}w'(\omega'_1,\omega_2)&=&\proba'_1(\omega'_1)
\eea
\label{eliuthleiuty}
\eeq
Observe that $\sum_{\omega'_1}w'(\omega'_1,\omega_2)=\sum_{\omega_1}\sum_{\omega'_1\in\omega_1}w'(\omega'_1,\omega_2) = \sum_{\omega_1}w(\omega_1,\omega_2)=\proba_2(\omega_2)$ since $w(\omega_1,\omega_2)$ is the weighting function of the lifting $\system_1\NMPlift{\simu}\system_2$. Our claim is thus justified.

To solve (\ref{eliuthleiuty}), we observe that it splits into the following independent subsystems in which $\omega_1$ is seen as a parameter ranging over $\Omega_1$:
\beq
\bea{rrcl}
\forall\omega_2:&
\sum_{\omega'_1\in\omega_1}w'(\omega'_1,\omega_2)&=&w(\omega_1,\omega_2)
\\
[1mm]
\forall\omega'_1{\in}\omega_1:& \sum_{\omega_2}w'(\omega'_1,\omega_2)&=&\proba'_1(\omega'_1)
\eea
\label{lurhtliuehou}
\eeq
The rows of System (\ref{lurhtliuehou}) are linked by the following relation: summing over all $\omega_2$ the first set of equations yields $\sum_{\omega_2}\sum_{\omega'_1\in\omega_1}w'(\omega'_1,\omega_2)=\sum_{\omega_2}w(\omega_1,\omega_2)$ $=\proba_1(\omega_1)$, whereas 
summing over all $\omega'_1\in\omega_1$ the second set of equations yields
$\sum_{\omega'_1\in\omega_1}\sum_{\omega_2}w'(\omega'_1,\omega_2)=\sum_{\omega'_1\in\omega_1}\proba'_1(\omega'_1)=\proba_1(\omega_1)$, and the two resulting equations are identical, by Fubini theorem.

Let $K_2$ be the cardinal of $\Omega_2$ and $L_1$ the cardinal of the set \mbox{$\{\omega' _1\mid\omega'_1\in\omega_1\}$}.
We distinguish the three cases $L_1=1$, $K_2=1$, and $L_1,K_2>1$. 

If $L_1=1$, setting $\forall\omega_2:w'(\omega'_1,\omega_2){=}w(\omega_1,\omega_2)$ yields a solution to (\ref{lurhtliuehou}) since the last equation of (\ref{lurhtliuehou}) is trivially satisfied. 

Case $K_2=1$ is trivial either, since $w'(\omega'_1,\omega_2)=\proba'(\omega'_1)$ is the unique solution.

For the third case $L_1,K_2{>}1$, the system (\ref{lurhtliuehou}) has more unknowns ($K_2{\times}L_1$) than equations ($K_2{+}L_1$). To prove that it indeed has solutions, we reorganize the unknowns $w'(\omega'_1,\omega_2)$ into a row matrix by listing as a submatrix the $w'(\omega'_1,\omega_2)$ for every fixed value of $\omega'_1$ and $\omega_2$ ranging over $\Omega_2$:
\[\bea{l}
\hspace*{-1mm}
\left[w'(\omega'_{11},\omega_{21}),\dots,w'(\omega'_{11},\omega_{2K_2}),\right.
\\
w'(\omega'_{12},\omega_{21}),\dots,w'(\omega'_{12},\omega_{2K_2}),
\\
\hspace*{2.5cm} \vdots
\\
\left.w'(\omega'_{1L_1},\omega_{21}),\dots,w'(\omega'_{1L_1},\omega_{2K_2})\right]
\eea
\]
We arrange the equations as indicated in (\ref{lurhtliuehou}): we put on top the $K_2$ equations parameterized by $\omega_2$ followed by the $L_1$ equations parameterized by $\omega'_1$. For $A$ and $A'$ two matrices, of respective sizes $m{\times}n$ and $m'{\times}n'$, we denote by $A\otimes{A'}$ their \emph{Kronecker product} obtained by replacing the $a_{ij}$ entry of $A$ by the matrix $a_{ij}.A'$, thus obtaining a matrix of size $(m{\times}m')\times(n{\times}n')$.
With these conventions and notations, the matrix of the linear system (\ref{lurhtliuehou}) takes the following form, where $\Id_m$ denotes the identity matrix of size $m{\times}m$:
\beq
M = 
\left[\bea{c}
[\overbrace{1 \dots 1}^{L_1\;{\rm times}}]\otimes{\Id}_{K_2}
\\ [4mm]
\Id_{L_1}\otimes[\underbrace{1 \dots 1}_{K_2\;{\rm times}}]
\eea\right], 
\label{erouhgltohj}
\eeq
of size $(K_2{+}L_1)\times(K_2{\times}L_1)$. The proof that the first two lines of (\ref{eruithpeuhpoh}) are satisfied rests on the two lemmas~\ref{lerfiuglo} and~\ref{rliguhli} below.

We move to the third line of (\ref{eruithpeuhpoh}). The conditions $w'(\omega'_1,\omega_2)>0$ and $\omega'_1\,\cons'_1\,{q_1}$ together imply 
 $w(\omega_1,\omega_2)>0$ and $\omega_1\,\cons_1\,{q_1}$ where $\omega_1$ is the equivalence class of $\omega'_1$, i.e., $\omega'_1\in\omega_1$. The right hand side then follows since we have $\system_1\NMPlift{\simu}\system_2$. This finishes the proof.
\end{proof}

\begin{lemma}
	\label{lerfiuglo} 
	If $L_1,K_2>1$, then the matrix $M$ defined in (\ref{erouhgltohj}) has row rank equal to $L_1+K_2-1$.
	
\end{lemma}
\begin{proof}
	We proceed by double induction over $L_1,K_2$. 
	The base case is $L_1{=}K_2{=}2$, for which matrix $M$ is equal to
	\[
	M=\left[\bea{cccc}
	\gemph{1} & \gemph{0} & \gemph{1} & \gemph{0} \\
	\blemph{0} & \remph{1} & {0} & {\gemph{1}} \\
	\remph{1} & \blemph{1} & {0} & \gemph{0} \\
	\blemph{0} & {0} & {\remph{1}} & \gemph{1}
	\eea
	\right]
	\]
	$M$ is singular but the submatrix obtained by erasing the first row and the last column in $M$ (the latter are shown in green)
	is regular. This is proved by observing that this submatrix possesses only one traversal,\footnote{A \emph{traversal} of a $p{\times}p$-matrix $B$ is a selection of $p$ non-zero entries of $B$ visiting all columns and rows of $B$.} shown in red, hence its determinant equals $\pm{1}$ and cannot be zero.
	
	In the rest of the proof, we use the convention that symbols written in boldface denote a matrix of suitable sizes filled with the indicated symbol. For example, $\bf 0$ denotes a matrix filled with zeros, the sizes of which depend on the context.
	
	For the induction argument, let $M(L_1,K_2)$ denote the matrix defined in (\ref{erouhgltohj}) with the values $L_1,K_2$ and $\overline{M}(L_1,K_2)$ the square submatrix of $M(L_1,K_2)$ obtained by erasing the first row in $M(L_1,K_2)$ and then selecting columns accordingly. Using these notations, the invariant of the induction argument is the following:
	\beq
	\mbox{
	 The number of traversals of $\overline{M}(L_1,K_2)$  equals  $1$.
}
	\eeq
	
	Increasing $L_1$ by $1$: matrix $M(L_1,K_2)$ becomes
	\beq
	M(L_1{+}1,K_2)= 
	\left[\bea{cc}	
	M(L_1,K_2) & \left[~\bea{c} \remph{\Id_{K_2}} \\ \\ \remph{\bf 0} 
	\eea~\right]
	\\ [7mm]
	\remph{\bf 0}  & \bigl[\underbrace{\remph{1 \dots 1}}_{K_2\;{\rm times}}\bigr]
	\eea\right]
	\label{gouhohjiotu}
	\eeq
	where the added part is highlighted in red.
	We construct $\overline{M}(L_1{+}1,K_2)$ by adding, to $\overline{M}(L_1,K_2)$, one row below and one among the $K_2$ new columns shown on the right part of $M(L_1{+}1,K_2)$. For this case the number of traversals keeps constant.
	
	Increasing $K_2$ by $1$: matrix $M(L_1,K_2)$ becomes
	\beqq
M(L_1,K_2{+}1) =
\left[\bea{c}
\bigl[\overbrace{1 \dots 1}^{L_1\;{\rm times}}\bigr]\otimes\left[\bea{ccc}
{\Id}_{K_2} && \remph{\bf 0} \\ \remph{\bf 0} && \remph{1}
\eea\right]
\\ [4mm]
\hspace*{9mm}
\Id_{L_1}\otimes\bigl[\underbrace{1 \dots 1}_{K_2\;{\rm times}} ~~ \remph{1}~~\bigr]
\eea\right]
\eeqq
where the additional entries are shown in red. We move the new row 
\[
[\overbrace{1 \dots 1}^{L_1\;{\rm times}}]\otimes[\remph{{\bf 0}~~ 1}]
\]
 to the last line of the matrix. The new columns arising from \[
[\overbrace{1 \dots 1}^{L_1\;{\rm times}}]\otimes\left[\bea{c}
\remph{\bf 0} \\ \remph{1} \\ \remph{1} \eea
\right]
\]
are all shifted to the right to become the last ones of the matrix while keeping the same order. Having done this, we end up with a reorganized matrix that has the following form:
	\beqq
	M(L_1,K_2{+}1)= 
	\left[\bea{cc}	
	M(L_1,K_2) & \left[~\bea{c} ~\,\remph{\star} ~\,
	\eea~\right]
	\\ [3mm]
	\remph{\bf 0}  & \bigl[\underbrace{\remph{1 \dots 1}}_{K_2\;{\rm times}}\bigr]
	\eea\right]
	\eeqq
	where the added part is highlighted in red. Again the number of traversals remains constant.
\end{proof}

In the following, for $X$ a matrix, $X^T$ denotes its transpose. Also, we take the convention that vectors identify with column matrices.

\begin{lemma}
	\label{rliguhli} 
	Let $A$ be an $m{\times}n$ matrix with $m\leq{n}$ such that $A$ has rank $m-1$, and there exists a non-zero $m$-vector $v$ such that $v^TA=[0\dots0]$. Then, for every $m$-vector $y$ such that $v^Ty=0$, the linear system $Ax=y$ possesses a solution.
	
\end{lemma}
\begin{proof}
	We complete $v$ with $m-1$ vectors to get a basis of $\bR^m$ and denote by $C$ the $m{\times}m$-matrix obtained by taking this basis as its columns, $v$ being the first one. Premultiplying the linear system $Ax=y$ by $C^T$ yields $C^TAx=C^Ty$. Vector $C^Ty$ has a $0$ as its first entry, completed by an $m{-}1$-vector that we denote by $z$. Similarly, matrix $C^TA$ has its first row equal to zero, and we denote by $B$ the matrix obtained by erasing the first row of $C^TA$. Our original linear system is then equivalent to the reduced linear system $Bx=z$. By assumption, $B$ has rank $m{-}1$, i.e., full row rank, which ensures that a solution to $Bx=z$ exists (possibly not unique).
\end{proof}

To prove that the first two lines of (\ref{eruithpeuhpoh}) are satisfied, we apply Lemma~\ref{lerfiuglo} to the matrix $M$ defined in (\ref{erouhgltohj}), and then Lemma~\ref{rliguhli} to the matrix $M$ with 
\[
v^T=\bigl[\;
\underbrace{1~\dots ~1}_{K_2\;{\rm times}}~
\underbrace{{-}1~\dots ~{-}1}_{L_1\;{\rm times}}
\;\bigr]
\]
\finremovalbert{}

\subsection{Proof of Lemma~\ref{uweygfkiutygf}}
\label{guihepioru}
\removalbert{shift proof to appendix}

\begin{proof}
  By definition,
  $\system_1\left(\NMPlift{\simu_{12}}\reldot\NMPlift{\simu_{23}}\right)\system_3$
  iff there exists $\system_2\in\Systems(Q_2)$ such that
  $\system_1\NMPlift{\simu_{12}}\system_2$ and
  $\system_2\NMPlift{\simu_{23}}\system_3$, that is, there exists two
  weighted functions $w_{12}$ over $\Omega_1\times{\Omega_2}$
  and $w_{23}$ over $\Omega_2\times{\Omega_3}$, such that
\begin{itemize}
	\item $w_{12}$ projects to $\proba_1$ and $\proba_2$, and $w_{23}$ projects to $\proba_2$ and $\proba_3$, and 
  \item 
  $w_{12}(\omega_1,\omega_2)>0$ and $\omega_1\,\cons_1\,{q_1}$ together  imply the existence of a $q_2$ such that 
$\omega_2\,\cons_2\,{q_2}$ and $q_1\,\simu_{21}\,{q_2}$; 
  
  $w_{23}(\omega_2,\omega_3)>0$ and $\omega_2\,\cons_2\,{q_2}$ together  imply the existence of a $q_3$ such that 
$\omega_3\,\cons_3\,{q_3}$ and $q_2\,\simu_{23}\,{q_3}$. 
\end{itemize}
  On the other hand,   $\system_1\NMPlift{(\simu_{12}\reldot\simu_{23})}\system_3$ iff there
  exists a weighted function $w$ over $\Omega_1{\times}\Omega_3$ projecting to
  $\proba_1$ and $\proba_3$ and such that: $w(\omega_1,\omega_3)>0$ and $\omega_1\,\cons_1\,{q_1}$ together imply the existence of a $q_3$ such that 
$\omega_3\,\cons_3\,{q_3}$ and  $q_1\,(\simu_{12}\reldot\simu_{23})\,{q_3}$.
  
  We thus construct the following function $w$
  defined over $\Omega_1\times{\Omega_3}$:
\beq
w(\omega_1,\omega_3) = \sum_{\omega_2\in\Omega_2} 
{w_{12}(\omega_1,\omega_2).w_{23}(\omega_2,\omega_3)}
\label{ltgiudtlu}
\eeq
To show that
$\system_1\NMPlift{(\simu_{12}\reldot\simu_{23})}\system_3$, we have to prove
the following regarding $w$:
\begin{itemize}
\item if $w(\omega_1,\omega_3) > 0$ and $\omega_1\,\cons_1\,{q_1}$ hold, then we can find $q_3$ such that 
$\omega_3\,\cons_3\,{q_3}$ and $q_1~(\simu_{12}\reldot\simu_{23})~q_3$. To show this, note that if
  $w(\omega_1,\omega_3) > 0$ then, by (\ref{ltgiudtlu}),  we can find an $\omega_2$ such that $w_{12}(\omega_1,\omega_2){.}w_{23}(\omega_2,\omega_3)>0$. Since $w_{12}(\omega_1,\omega_2)>0$, there exists some $q_2$ such that 
  $\omega_2\cons_2{q_2}$ and $q_1\simu_{12}q_2$. Since $w_{23}(\omega_2,\omega_3)>0$, there exists some $q_3$ such that 
  $\omega_3\cons_3{q_3}$ and $q_2\simu_{23}q_3$. 
Now, we have  $q_1~\simu_{12}~q_2$
  and $q_2~\simu_{23}~q_3$, which implies $q_1~(\simu_{12}\reldot\simu_{23})~q_3$;
  
\item $w$ projects to $\proba_3$:
\beqq
\sum_{\omega_1} w(\omega_1,\omega_3) & = & \sum_{\omega_1}\sum_{\omega_2} 
{w_{12}(\omega_1,\omega_2){.}w_{23}(\omega_2,\omega_3)}
\\
\mbox{by Fubini}
 & = & \sum_{\omega_2}\sum_{\omega_1} 
 {w_{12}(\omega_1,\omega_2){.}w_{23}(\omega_2,\omega_3)}
\\
 & = & \sum_{\omega_2} 
{w_{23}(\omega_2,\omega_3)}
\underbrace{\sum_{\omega_1}w_{12}(\omega_1,\omega_2)}_{=1} \\
& = & \proba_3(\omega_3)
\eeqq

\item $w$ projects to $\proba_1$: this is proved similarly.
\end{itemize}
Therefore,
$\system_1\left(\NMPlift{\simu_{12}}\reldot\NMPlift{\simu_{23}}\right)\system_3$
iff $\system_1\NMPlift{(\simu_{12}\reldot\simu_{23})}\system_3$.
\end{proof}
\finremovalbert{}

\section{Proofs regarding MMDPs}
\subsection{Proof of Lemma~\ref{glrtukghtrllsdukfg}}
\label{wroifuwgopi}

\begin{proof}
Set $M'\eqdef M'_1\times{M'_2}$ and $M\eqdef M_1\times{M_2}$. Define the relation $\leq$ between $R'$ and $R$ by: $r'\leq{r}$ iff
$r'_1\leq_1{r_1}$ and $r'_2\leq_2{r_2}$. Let us prove that $\leq$ is a
simulation. 

Let $r'$ be such that $\trans{r'}{\action}{\system'}{M'}$ for some consistent $\system'$. Then, 
$r'=r'_1\join{r'_2}$ and $\system'=\system'_1\times\system'_2$. By
definition of the parallel composition, we have
$\trans{r'_i}{\action}{\system'_i}{M'_i}$ for $i=1,2$. Since
$r'_i\leq{r_i}$, we derive the existence (and uniqueness) of
consistent systems $S_i,i=1,2$ such that
$\trans{r_i}{\action}{\system_i}{M_i}$. Since $r=r_1\join{r_2}$ we
have $r_1\compat{r_2}$ and, thus, by definition of the parallel
composition, we deduce
\mbox{$\trans{r}{\action}{\system_1\times\system_2}{M}$}. 

It remains to show that $\system_1\times\system_2$ is consistent. To
prove this, remember that $\system'=\system'_1\times\system'_2$ is
consistent. Thus, there exist compatible $r'_1$ and $r'_2$ such that
$\produces{\system'_i}{r'_i}, i=1,2$.  By definition of the
simulations $\leq_i$, we deduce that
$\produces{\system_i}{r_i}, i=1,2$, which shows that
$\system_1\times\system_2$ is consistent.
\end{proof}

\section{Proofs regarding Probabilistic Automata}
\subsection{Proof of Theorem~\ref{erlgfuierhlpiu}}

Defining simulation relations for $\pa$ requires lifting relations,
from states to distributions over states. The formal definition for
this lifting, as given in Section\,4.1 of~\cite{Seg06-CONCUR},
corresponds to our Definition~\ref{hrgfuihsk}, when restricted to
purely probabilistic mixed systems.

The same holds for the strong simulation relation defined in
Section\,4.2 of the same reference: it is verbatim our
Definition~\ref{def:simulation}, when restricted to purely
probabilistic mixed systems. This proves the part of
Theorem~\ref{erlgfuierhlpiu} regarding simulation.

We move to parallel composition, for which the reader is referred to~\cite{DBLP:conf/concur/LynchSV03}, Section 3. For $P_1=(\alphabet,Q_1,q_{0,1},\ra_1)$ and $P_2=(\alphabet,Q_2,q_{0,2},\ra_2)$ two PA, their parallel composition is  $P=P_1\times{P_2}=(\alphabet,Q_1\times{Q_2},(q_{0,1},q_{0,2}),\ra)$, where 
\beq
\trans{(q_1,q_2)}{\action}{\proba_1 {\otimes} \proba_2}{}
 &\mbox{ iff }&
 \trans{q_i}{\action}{\proba_i}{i} \mbox{ for }i=1,2 
\label{rtogtrilko}
\eeq
So, on one hand we consider the \mmdp\ $M_P$.
On the other hand, we consider the parallel composition of the mappings $M_{P_1}$ and $M_{P_2}$, that is $M=M_{P_1}\times M_{P_2}=(\alphabet,\{\xi_1,\xi_2\},(q_{0,1},q_{0,2}),\ra_{12})$,
so that the state space is the domain of the pair $(\xi_1,\xi_2)$, namely $Q_1\times{Q_2}$, 
and, since there is no shared variable between the two \mmdp, the transition relation $\ra_{12}$ is given by:
\beq
\trans{(q_1,q_2)}{\action}{\system_1 {\times} \system_2}{12}
&\mbox{ iff }& \trans{q_i}{\action}{\system_i}{i} \mbox{ for }i=1,2
\label{toguihtuiokuygf}
\eeq
We thus need to show that 
\beq
\mbox{$M_P$ and $M$ are simulation equivalent.}
\label{gtiohrgtio}
\eeq
We will actually show that the identity relation between the two state spaces (both are equal to $Q_1\times{Q_2}$) is a simulation relation in both directions.

Observe first that (\ref{rtogtrilko}) and (\ref{toguihtuiokuygf}) differ in that the former involves a nondeterministic transition relatiobn, whereas the latter involves a deterministic transition function, mapping states to mixed systems.

Pick $(q_1,q_2)\in{Q_1}\times{Q_2}$ and consider a transition for $M_P$:
\[
\trans{(q_1,q_2)}{\action}{S}{M_P}=((\Omega,\Proba),\xi,(q_{0,1},q_{0,2}),\cons)
\]
where we have, for $S$:
\begin{itemize}
	\item $\Omega$ is the product of $n_1$ copies of $Q_1$ and $n_2$ copies of $Q_2$, where, for $i=1,2$, $n_i$ is the cardinality of the set $\{\proba_i\mid(q_i,\action,\proba_i)\in\ra_i\}$, so that $\omega$ identifies $n_1\times{n_2}$-tuple of states: $\omega=(q_{11},\dots,q_{1n_1};q_{21},\dots,q_{2n_2})$;
	\item $\Proba$ is the product of all probabilities belonging to set $$\{\proba_1\otimes\proba_2\mid(q_i,\action,\proba_i)\in\ra_i\}$$
	\item $\xi$ has domain $Q_1\times{Q_2}$;
	\item $(\omega,(q_1,q_2))\in\cons$ if and only if \[
	(q_1,q_2)\in\{(q_{1i_1},q_{2i_2})\mid i_1\in\{1,\dots,n_1\} \mbox{ and } i_2\in\{1,\dots,n_2\}\}\,.
	\]
\end{itemize}
Next, pick $(q_1,q_2)\in{Q_1}\times{Q_2}$ and consider a transition for $M$, see (\ref{toguihtuiokuygf}). We need to detail what $S_1\times{S_2}=((\Omega' ,\Proba' ),\xi' ,(q'_{0,1},q'_{0,2}),\cons')$ is. We have, for $S_1\times{S_2}$:
\begin{itemize}
	\item $\Omega'$ is still the product of $n_1$ copies of $Q_1$ and $n_2$ copies of $Q_2$;
	\item $\Proba'$ is the product $\Proba_1\otimes\Proba_2$, where $\Proba_i$ is the product of all probabilities belonging to set \mbox{$\{\proba_i\mid(q_i,\action,\proba_i)\in\ra_i\}$};
	\item $\xi'$ has domain $Q_1\times{Q_2}$;
	\item $(\omega,(q_1,q_2))\in\cons'$ if and only if \[
	(q_1,q_2)\in\{(q_{1i_1},q_{2i_2})\mid i_1\in\{1,\dots,n_1\} \mbox{ and } i_2\in\{1,\dots,n_2\}\}\,.
	\]
\end{itemize}
By associativity of $\otimes$, $\Proba'=\Proba$, whereas other items for $S$ on the one hand and other items for $S_1\times{S_2}$ on the other hand, are synctatically identical. Thus (\ref{gtiohrgtio}) follows.

\section{Proofs regarding Mixed Interfaces}
\subsection{Proof of Lemma~\ref{lemma-cleaning}}
\label{erpguioehpguio}
\begin{proof}
  \begin{itemize}
  \item We remove from $\contract$ inconsistent states $q$;
\begin{itemize}
\item if $\ntransmayindex{q}{\action}{}{\contract}$ and
  $\transmustindex{q}{\action}{\Systems^\probamust}{\contract}$ then
  $q$ cannot be involved in a simulation relation allowing to state
  that $M$ is a model of $\contract$ because of (\ref{eq-modele-may}) in
  the definition of the model relation.
\item if $\transmustindex{q}{\action}{}{\contract}$ and
  $\transmayindex{q}{\action}{\Systems^\probamay}{\contract}$ but
  $\Systems^\probamust\cap\Systems^\probamay$ contains no consistent
  system in the sense of Definition~\ref{slergiuhpiu} $q$ cannot be
  involved in a simulation relation allowing to state that $M$ is a
  model of $\contract$ because of (\ref{eq-modele-must}) in the
  definition of the model relation.
   \end{itemize}
   As a result, $q$ plays no role in the semantics of $\contract$ and
   its lack in $\prune{\contract}$ does not change the semantics.
 \item We remove from $\contract$ some may transitions to inconsistent
   states which could not be realized by any model $\contract$. \eproof
  \end{itemize}
\end{proof} 

\subsection{Proof of Lemma~\ref{eorgfuirhpu}}
\label{guioenhuio}
\begin{proof}
  The reflexity of $\refines$ follows immediately from
  Definition~\ref{def-refinement}.

  Now for the transitivity, assume that
  $\contract_1 \refines \contract_2$ and
  $\contract_2 \refines \contract_3$. with the respective refinement
  relations $\refines_{12}\;\subseteq\;{Q_1}\times{Q_2}$ and
  $\refines_{23}\;\subseteq\;{Q_2}\times{Q_3}$.

  Define now using notation (\ref{rbgjgnwoun}): \beq
  \refines_{13}&=&\refines_{12}\reldot\refines_{23}\,.
\label{comprelref}
\eeq

Let $q_1$ and $q_3$ such that $q_1\refines_{13}q_3$. By
\ref{comprelref}, we have $q_1\refines_{12}q_2$ and
$q_2\refines_{23}q_31$ for some $q_2$. Thus, for all $\action$ such
that $\transmayindex{q_1}{\action}{\Systems^\probamay_1}{1}$, we have
$\transmayindex{q_2}{\action}{\Systems^\probamay_2}{2}$  and
${\Systems^\probamay_1}\subseteq^{\refines_{12}}\,{\Systems^\probamay_2}$. Moreover, $\transmayindex{q_3}{\action}{\Systems^\probamay_3}{3}$  and
${\Systems^\probamay_2}\subseteq^{\refines_{23}}\,{\Systems^\probamay_3}$. By
Lemma~\ref{uweygfkiutygf}, we have
${\Systems^\probamay_1}\subseteq^{\refines_{13}}\,{\Systems^\probamay_3}$.

Similarly for must transitions, for all $\action$ such
that $\transmustindex{q_3}{\action}{\Systems^\probamust_3}{3}$, we have
$\transmustindex{q_2}{\action}{\Systems^\probamust_2}{2}$  and
${\Systems^\probamust_2}\subseteq^{\refines_{23}}\,{\Systems^\probamust_3}$. Moreover, $\transmustindex{q_1}{\action}{\Systems^\probamust_1}{1}$  and
${\Systems^\probamust_1}\subseteq^{\refines_{12}}\,{\Systems^\probamust_2}$. By
Lemma~\ref{uweygfkiutygf}, we have
${\Systems^\probamust_1}\subseteq^{\refines_{13}}\,{\Systems^\probamust_3}$.
A a result, we have $\contract_1 \refines \contract_3$.\eproof
\end{proof}

\subsection{Proof of Theorem~\ref{oer8w5t7ho0u}}
\label{eluihuihipgh}
\begin{proof}
  Assume
  $\contract_1\refines\contract_2$ and consider the refinement relation $\refines\;\subseteq\;{Q_1}\times{Q_2}$.
Let $M$ be a model of $\contract_1$ and let $(r,q_1)\in{R}\times{Q_1}$ 
satisfy 
$r{\models_1}\,q_{1}$.
Focus first on the \emph{may} transition relation.  By
(\ref{eq-modele-may}) applied to $\models_1$, for any $\action$ such that
\mbox{$\transindex{r}{\action}{\system_M}{M}$}
\beq
\mbox{$\transmayindex{q_{1}}{\action}{\Systems^\probamay_1}{\contract_1}$ and
$\system_M\in^{\models_1}{\Systems^\probamay_1}$ both hold.} 
\label{gtrjuhliu}
\eeq
Let $q_2\in{Q_2}$ be such that $q_1\refines{q_2}$. 
Using the first condition of (\ref{507tho5u0857}), we get
\beq
\mbox{$\transmayindex{q_2}{\action}{\Systems^\probamay_2}{2}$ and
${\Systems^\probamay_1}\subseteq^\refines\,{\Systems^\probamay_2}$}
\label{vuwguigqlig}
\eeq
Define  the relation:
$r\models_2{q_2} \Leftrightarrow \exists q_1\in{Q_1} : 
r\models_1{q_1}  \mbox{ and }
q_1\refines{q_2}$.
Using notation (\ref{rbgjgnwoun}), we have 
\beq
\models_2&=&\models_1\reldot\refines\,.
\label{elfiyubgpi}
\eeq
Now, let $q_2$ be such that $r\models_2{q_2}$. Combining (\ref{gtrjuhliu}) and (\ref{vuwguigqlig}) yields
\beq
\transmayindex{q_2}{\action}{\Systems^\probamay_2}{2} &\mbox{and}&
\system_M\in^{\models_1}{\Systems^\probamay_1}\subseteq^\refines{\Systems^\probamay_2}
\label{trhlyknmhjtokij}
\eeq
which, by (\ref{elfiyubgpi}) and Lemma~\ref{uweygfkiutygf}, yields
$\system_M\in^{\models_2}{\Systems^\probamay_2}$.
Combining this and (\ref{trhlyknmhjtokij}) shows that $r\models_2{q_2}$.
Focus next on the \must\ transition relation. Since $M$ is a model of $\contract_1$, 
(\ref{eq-modele-must}) applied to $\models_1$ yields the existence of $\system_M\in^{\models_1}\Systems^\probamust_1\subseteq^\refines{\Systems^\probamust_2}$ such that $\transindex{r}{\action}{\system_M}{M}$, which implies that (\ref{eq-modele-must})  holds for $\models_2$ by the same reasoning as before.
\end{proof}

\subsection{Proof of Theorem~\ref{kuygpgiohtoi}}
\label{erouigheguio}
\begin{proof}
  Using Theorem~\ref{oer8w5t7ho0u}, the previous statements
  follow from $
%
\contract_1 \wedge \contract_2 \refines \contract_i \mbox{ for } i=1,2$.
%
Take the first projection as the candidate refinement relation, namely: $(q_1,q_2)\refines{q_1}$ for $(q_1,q_2)$ and $q_1$ reachable from their respective initial states. Using the four rules of Definition~\ref{def-conjunction}, we get $\contract_1\underline{\wedge}\contract_2\refines\contract_1$, and thus
$\contract_1 \wedge
\contract_2=\consistent{\contract_1\underline{\wedge}\contract_2}\refines\contract_1$
since $\contract_1$ possesses no inconsistent state. The same holds
for $\contract_2$ by symmetry.
\end{proof}

\subsection{Proof of Theorem~\ref{rtpery9uy95e8}}
\label{epiougtrhpiough}
\begin{proof}
We successively prove the two statements. Regarding 
Statement~\ref{jiytdfiout}), the same proof holds as for associativity and commutativity of the conjunction.
Regarding Statement~\ref{wieytdfiwet}), 
Property (\ref{lrtuhlup}) is an immediate consequence of Definitions~\ref{rtghltughtui},~\ref{def-model} and~\ref{g07hortgig}.
Focus next on (\ref{rltiugyhpe}).
  Assume
 \beqq
\bea{rcr} 
\transmayindex{(q,q_2)}{\action}{\,\Systems^\probamay{\times}\Systems^\probamay_2}{\contract\cpara\contract_2}
\\ 
\transmustindex{(q,q_1)}{\action}{\,\Systems^\probamust{\times}\Systems^\probamust_1}{\contract\cpara\contract_1}
\eea
\label {ethhjytujkt}
\eeqq
By the rules of the composition, we deduce that the premises of  (\ref{507tho5u0857}) holds, so we can apply rule  (\ref{507tho5u0857}) since $\contract_1 \refines
  \contract_2$, which yields
 \beqq
\bea{l} 
\transmayindex{q_1}{\action}{\Systems^\probamay_1}{1}
\mbox{ and } \ {\Systems^\probamay_2}\subseteq^\refines\,{\Systems^\probamay_1} 
\\
\transmustindex{q_2}{\action}{\Systems^\probamust_2}{2} 
\mbox{ and } \ {\Systems^\probamust_2}\subseteq^\refines\,{\Systems^\probamust_1}
\eea
\label {ghpovkopbm}
\eeqq
which implies 
 \beqq
\bea{l} 
\transmayindex{(q,q_1)}{\action}{\Systems^\probamay{\times}\Systems^\probamay_1}{\contract\cpara\contract_1}
\mbox{ and } \ {\Systems^\probamay{\times}\Systems^\probamay_2}\subseteq^{\refines'}\,{\Systems^\probamay{\times}\Systems^\probamay_1}
\\ 
\transmustindex{(q,q_2)}{\action}{\Systems^\probamust{\times}\Systems^\probamust_2}{\contract\cpara\contract_2} 
\mbox{ and } \ {\Systems^\probamust{\times}\Systems^\probamust_2}\subseteq^{\refines'}\,{\Systems^\probamust{\times}\Systems^\probamust_1}
\eea
\label {ghpovkopbm}
\eeqq
where $\refines'$ is defined by $(q,q_2)\refines'(q,q_1)$ iff $q_2{\refines}q_1$. This shows that $\refines'$ is a refinement.
\end{proof}

\section{Proofs regarding CMC}
\subsection{Proof of Lemma~\ref{elriueoiufhyi}}
\label{eroiuhpiaeug}
\begin{proof}
Let us detail the semantics of mixed system $S$, see Definition~\ref{slergiuhpiu}. First, we draw $r'\in\Omega=R$ according to the probability $\Proba(r,.)$: this corresponds to the drawing of the next state in Markov Chain $\mc$. Second, we nondeterministically select $(r'',v(r''))$ in the state space $R{\times}2^A$ of $S$ so that $(r',r'',v(r''))\in\cons$. The only solution is $(r',r',v(r'))$, which provides us with the second component $v(r')$ of the state. The two semantics coincide.\eproof
\end{proof}

\subsection{Proof of Lemma~\ref{eruiehrpiufhio}}
\label{eoruifygouiy}
\begin{proof}
To the satisfaction relation $\simu\subseteq{R}{\times}{Q}$ following Definition~\ref{ywtfqwkfqkuqf}, we associate the relation 
$\models_\simu\ \subseteq\ (R{\times}2^{A_\mc})\times(Q{\times}{2^{A_\cmc}})$, 
 defined by
\beq
(r,\bar{r})\;\models_\simu\;(q,\bar{q})
&\mbox{ iff }&\left\{\bea{l}
r\ \simu\ q \\
\bar{r}=v_\mc(r) \\
\bar{q}=v_\cmc(q) \\
\restrict{A_\cmc}{\bar{r}}=\bar{q}
\eea\right.
\label{eltuigherpguio}
\eeq
Observe that, vice versa, we recover $\simu$ from $\models_\simu$ by keeping only the first condition of it.
We have to prove that
\beq
\mbox{
\begin{minipage}{11cm}
	 $\simu$ is a satisfaction relation for CMC if and only if $\models_\simu$ is a satisfaction relation for \mmi.
\end{minipage}
} \label{duyetfutuyt}
\eeq

\paragraph*{We first prove the ``only if'' part of (\ref{duyetfutuyt}) }Let $(r,q)$ satisfy $r\,\simu\;q$.
By (\ref{eltuigherpguio}), $\models_\simu$ is a relation between the states of \mmdp\ $M_\mc$ and \mmi\ $\contract_\cmc$. With reference to Definition~\ref{def-model},  to show that $\models_\simu$ is a satisfaction relation, it is enough to show that only \emph{may} transitions of $\contract_\cmc$ are allowed for $M_\mc$---the condition related to the \emph{must} transitions is vacuously satisfied.

Let $(r,\bar{r})\,\models_\simu\,(q,\bar{q})$ and $\transindex{(r,\bar{r})}{}{\system_\mc}{M_\mc}$, where $\system_\mc=((\Omega_\mc,\proba_\mc),X_\mc,\cons_\mc)$ is defined by applying (\ref{ewyjtdfwkdtf}) to $M_\mc$. We must prove that the latter transition is allowed by the \emph{may} transitions of \mmi\ $\contract_\cmc$, i.e., the target mixed system $\system_\mc$ satisfies condition (\ref{eq-modele-may}), meaning that 
\beq
\mbox{
\begin{minipage}{5.2cm}
	 $\transmayindex{(q,\bar{q})}{}{\Systems_\cmc}{\contract_\cmc}$ and there exists $\system_\cmc\in\Systems_\cmc$ such that $\system_\mc\NMPlift{\models_\simu}\system_\cmc$.
\end{minipage}
}
\label{eroihoeifjjydtfyutf}
\eeq
To construct a mixed system $\system_\cmc$ satisfying (\ref{eroihoeifjjydtfyutf}), we start from $r\,\simu\,q$, which provides us with a transition sub-probability $\Delta$ satisfying the conditions~\ref{detfwiyt}) of Definition~\ref{ywtfqwkfqkuqf}. We then consider the mixed system $\system_\cmc=((\Omega_\cmc,\proba_\cmc),X_\cmc,\cons_\cmc)$, where:
\begin{itemize}
	\item $\Omega_\cmc=Q$;
	\item $\proba_\cmc=\Proba(r,.)\Delta$, which belongs to $\varphi(q)$ by Definition~\ref{ywtfqwkfqkuqf};
	\item $\cons_\cmc\subseteq\Omega_\cmc{\times}({Q}{\times}2^{A_\cmc})$ consists of the triples $(q',(q',\bar{q}'))$, where $q'$ ranges over $Q$, $\bar{q}'=v_\cmc(q')$, and $v_\cmc$ relates to $v_\mc$ by $v_\cmc(q')=\restrict{A_\cmc}{v_\mc(r')}$. By Condition~\ref{dtyuywtdf}) of Definition~\ref{ywtfqwkfqkuqf}, we get $\restrict{A_\cmc}{v_\mc(r')}\in{V_\cmc}(q')$.
\end{itemize}
Let us prove that the so constructed mixed system $\system_\cmc$ satisfies $\system_\mc\NMPlift{\models_\simu}\system_\cmc$. We must find a weighting function $w:R{\times}Q\ra[0,1]$ satisfying the conditions of Definition~\ref{hrgfuihsk}. We claim that the wanted weighting function is
\[
w(r',q')=\Proba(r,r')\Delta(r',q')\,.
\]
We now prove that Conditions~\ref{sggouigh}) and~\ref{leiurlyui}) of Definition~\ref{hrgfuihsk} are satisfied by $w$. We begin with Condition~\ref{leiurlyui}). We have $\sum_{r'}w(r',q')=\sum_{r'}\Proba(r,r')\Delta(r',q')=\Proba(r,.)\Delta(q')$ using (\ref{wcuetrcde}). On the other hand, 
\[
\sum_{q'}w(r',q')=\sum_{q'}\Proba(r,r')\Delta(r',q')=\Proba(r,r')\sum_{q}\Delta(r',q')=\Proba(r,r')
\]
 by Condition~\ref{uetydyut} of Definition~\ref{ywtfqwkfqkuqf}. 

Focus next on Condition~\ref{sggouigh}) of Definition~\ref{hrgfuihsk}.
Pick $(r',q';(r',v_\mc(r'))$ such that $w(r',q'){>}0$, which implies $\Delta(r',q'){>}0$. Then by Condition~\ref{ioucdgfouy} of Definition~\ref{ywtfqwkfqkuqf}, $r'\simu\,q'$ holds. On the other hand, we have $(q',q',\restrict{A_\cmc}{v_\mc(r')})\in\cons$, showing that $(q',\restrict{A_\cmc}{v_\mc(r')})$ is the  state of $S_\cmc$ wanted in Condition~\ref{sggouigh}) of Definition~\ref{hrgfuihsk}. Hence, the so constructed mixed system $\system_\cmc$ satisfies $\system_\mc\NMPlift{\models_\simu}\system_\cmc$.
This proves the ``only if''  part of (\ref{duyetfutuyt}). 

\paragraph*{We now move to the ``if'' part of (\ref{duyetfutuyt})}
Let $(r,\bar{r})\,\models_\simu\,(q,\bar{q})$. Then by the definition (\ref{eltuigherpguio}) of relation $\models_\simu$, we deduce that $r\,\simu\,q$ holds and we must prove that $\simu$ is a satisfaction relation for CMC. To this end we use the fact that $\models_\simu$ is a satisfaction relation for \mmi, namely: if $\transindex{(r,\bar{r})}{}{\system_\mc}{M_\mc}$, then there exists $\system_\cmc\in\Systems_\cmc$ such that $\system_\mc\NMPlift{\models_\simu}\system_\cmc$. The target system $\system_\cmc$ takes the form $\system_\cmc=((\Omega,\proba),X,\cons)$, where:
\begin{itemize}
	\item $\Omega=Q$;
	\item $\proba(q')=\sum_{r'\in{R}}w(r',q')$, where $w(r',q')$ is the weighting function associated to the lifting of relation $\models_\simu$;
	\item $\cons\subseteq\Omega{\times}({Q}{\times}2^{A_\cmc})$ consists of the triples of the form $(q',q',\restrict{A_\cmc}{v_\mc(r')})$, where $r'$ ranges over $R$ and $r'\,\simu\,q'$. 
\end{itemize}
In proving that the relation $\simu$ inferred from $\models_\simu$ is a satisfaction relation for CMC, we must find the $\Delta$ occurring in  Definition~\ref{ywtfqwkfqkuqf}. We define it as
\[
\Delta(r',q') = \left\{\bea{l}\displaystyle\frac{w(r',q')}{\Proba(r,r')} \mbox{ if } \Proba(r,r')>0 \\ [3mm] 0 \mbox{ otherwise.}
\eea\right.
\]
The conditions of Definition~\ref{ywtfqwkfqkuqf} are satisfied. This finishes the proof of the ``if'' part and the lemma is proved.
\end{proof}

\subsection{Proof of Lemma~\ref{elriuhepiu}}
\label{erlivgfuebhiu}
\begin{proof}
The proof follows the same lines as for Lemma~\ref{eruiehrpiufhio}. To the refinement relation $\simu\subseteq{Q_2}{\times}{Q_1}$ following Definition~\ref{dweiytdfutdrfyf}, we associate the relation 
\[
\refines_\simu\ \subseteq\ (Q_2{\times}{2^{A_{2}}})\times(Q_1{\times}{2^{A_{1}}})
\]
 defined by
\beq
(q_2,\bar{q}_2)\;\refines_\simu\;(q_1,\bar{q}_1)
&\mbox{ iff }&\left\{\bea{l}
q_2\ \simu\ q_1 \\
\bar{q}_2=v_{2}(q_2) \\
\bar{q}_1=v_{1}(q_1) \\
\bar{q}_1=\bar{q}_{2_{\left\downarrow{A_1}\right.}}
\eea\right.
\label{oerigfyubrpiu}
\eeq
By (\ref{oerigfyubrpiu}), $\refines_\simu$ is a relation between the states of \mmi\ $\contract_{\cmc_2}$ and  $\contract_{\cmc_1}$.
Observe that, vice versa, we recover $\simu$ from $\refines_\simu$ by keeping only the first condition of it. We have to prove that
\beq
\mbox{
\begin{minipage}{9cm}
	 $\simu$ is a weak refinement relation for CMC if and only if $\refines_\simu$ is a refinement relation for \mmi{s}.
\end{minipage}
} \label{erigfuiuygf}
\eeq

\paragraph*{We first prove the ``only if'' part of (\ref{erigfuiuygf})} Let $(q_2,q_1)$ satisfy $q_2\,\simu\;q_1$. 
With reference to Definition~\ref{def-refinement}, to show that $\refines_\simu$ is a refinement relation, it is enough to show the first condition of (\ref{507tho5u0857})---the condition related to the \emph{must} transitions is vacuously satisfied.	

From $(q_2,\bar{q}_2)\;\refines_\simu\;(q_1,\bar{q}_1)$ and $\transmayindex{q_2}{}{\Systems^\probamay_2}{2}$, we have to deduce  
\[
\transmayindex{q_1}{}{\Systems^\probamay_1}{1}
\mbox{ and } \ {\Systems^\probamay_2}\subseteq^{\refines_\simu}{\Systems^\probamay_1}\,,
\]
which translates as
\beq
\mbox{
\begin{minipage}{5cm}
	for every $S_2\in\Systems^\probamay_2$ we can find  $S_1\in\Systems^\probamay_1$ such that $\system_2\refines_\simu^\Systems\system_1$.
\end{minipage}
}
\label{fwiytedfwiuwtf}
\eeq
Let $\system_{2,v}$ have the form $\system_{2,v}=((\Omega_2,\proba_2),X_2,\cons_{2,v})$ following (\ref{uytdfiytfiyt}).
To construct a mixed system $\system_1$ satisfying (\ref{fwiytedfwiuwtf}) we start from $q_2\,\simu\;q_1$, which provides us with a transition sub-probability $\Delta$ satisfying the Conditions~\ref{idfgkukeyufgr}) of Definition~\ref{dweiytdfutdrfyf}. We then consider the mixed system $\system_1=((\Omega_1,\proba_1),X_1,\cons_1)$, where:
\begin{itemize}
	\item $\Omega_1=Q_1$;
	\item $\proba_1=\proba_2\Delta$, which belongs to $\varphi_1(q_1)$ by Definition~\ref{dweiytdfutdrfyf};
	\item $\cons_1\subseteq\Omega_1{\times}({Q_1}{\times}{2^{A_{1}}})$ consists of the triples of the form $(q'_1,q'_1,\restrict{A_{1}}{v(q'_2)})$, where $v$ is the one arising in the definition of $\system_{2,v}$ and $q'_1$ ranges over $Q_1$. By Condition~\ref{ieufgierufyg}) of Definition~\ref{dweiytdfutdrfyf}, we get $\restrict{A_{1}}{v(q'_2)}\subseteq{V_{1}}(q'_1)$.
\end{itemize}
Let us prove that the mixed system $\system_1$ satisfies $\system_{2,v}\NMPlift{\refines_\simu}\system_1$. We must find a weighting function $w:Q_2{\times}Q_1\ra[0,1]$ satisfying the conditions of Definition~\ref{hrgfuihsk}. We claim that the wanted weighting function is
\[
w(q'_2,q'_1)=\proba_2(q'_2)\Delta(q'_2,q'_1)\,.
\]
Let us prove that Conditions~\ref{sggouigh}) and~\ref{leiurlyui}) of Definition~\ref{hrgfuihsk} are satisfied by $w$. We begin with Condition~\ref{leiurlyui}). We have $\sum_{q'_2}w(q'_2,q'_1)=\sum_{q'_2}\proba_2(q'_2)\Delta(q'_2,q'_1)=\proba_1(q'_1)$ by definition of $\proba_1$. On the other hand, 
\[
\sum_{q'_1}w(q'_2,q'_1)=\sum_{q'_1}\proba_2(q'_2)\Delta(q'_2,q'_1)=\proba_2(q'_2)
\]
 by Condition~\ref{efgeorufygeu} of Definition~\ref{dweiytdfutdrfyf}. 

Focus next on Condition~\ref{sggouigh}) of Definition~\ref{hrgfuihsk}.
Pick $(q'_2,q'_1;(q'_2,v(q'_2))$ such that $w(q'_2,q'_1)>0$, which implies $\Delta(q'_2,q'_1)>0$. Then by Condition~\ref{geyugfwkugfwyt} of Definition~\ref{dweiytdfutdrfyf}, $q'_2\simu\,q'_1$ holds. On the other hand, we have $(q'_1,q'_1,\restrict{A_{1}}{v(q'_2)})\in\cons_1$, showing that $(q'_1,\restrict{A_{1}}{v(q'_2)})$ is the  state of $S_1$ wanted by Condition~\ref{sggouigh}) of Definition~\ref{hrgfuihsk}. Hence, the so constructed mixed system $\system_1$ satisfies $\system_{2,v}\NMPlift{\refines_\simu}\system_1$.
This proves the ``only if''  part of (\ref{erigfuiuygf}).

\paragraph*{We next move to the ``if'' part of (\ref{erigfuiuygf})}
Let $(q_2,\bar{q}_2)\,\refines_\simu\,(q_1,\bar{q}_1)$. Then by the definition (\ref{oerigfyubrpiu}) of relation $\refines_\simu$, we deduce that $q_2\,\simu\,q_1$ holds and we must prove that $\simu$ is a weak refinement relation for CMC. To this end we use the fact that $\refines_\simu$ is a modal refinement relation for \mmi, namely: if
 $\transmayindex{q_2}{}{\Systems^\probamay_2}{2}$, then  
$\transmayindex{q_1}{}{\Systems^\probamay_1}{1}
\mbox{ and } \ {\Systems^\probamay_2}\subseteq^{\refines_\simu}{\Systems^\probamay_1}$. 
That is, for any $\system_{2,v_2}\in\Systems^\probamay_2$, of the form  $\system_{2,v_2}=((\Omega_2,\proba_2),X_2,\cons_{2,v_2})$ following (\ref{uytdfiytfiyt}), there exists $\system_{1,v_1}=((\Omega_1,\proba_1),X_1,\cons_{1,v_1})\in\Systems^\probamay_1$ such that 
\beq
\system_{2,v_2}\NMPlift{\refines_\simu}\system_{1,v_1}\,.
\label{wuqdtrwdfuity}
\eeq
Condition (\ref{wuqdtrwdfuity}) and Definition~\ref{hrgfuihsk} of the lifting of a relation together imply the existence of a weighting function $w(q'_2,q'_1)$ satisfying the following conditions:
\begin{enumerate}
	\item For every triple $(q'_2,q'_1;(q'_2,v_2(q'_2))$ such that 
	\[
	w(q'_2,q'_1){>}0\mbox{ and }(q'_2,(q'_2,v_2(q'_2)){\in}\cons_{2,v_2}\,,
	\]
	there exists $(q'_1,v_1(q'_1))$ such that 
	\[
	(q'_1,(q'_1,v_1(q'_1))\in\cons_{1,v_1}\mbox{ and }
	(q'_2,\bar{q}'_2)\refines_\simu(q'_1,\bar{q}'_1)\,.
	\]
 
	\item $\sum_{q'_2}w(q'_2,q'_1)=\proba_1(q'_1)$ and $\sum_{q'_1}w(q'_2,q'_1)=\proba_2(q'_2)$. 
\end{enumerate}
In proving that the relation $\simu$ inferred from $\refines_\simu$ is a weak refinement relation for CMC, we must find the $\Delta$ occurring in  Definition~\ref{dweiytdfutdrfyf}. We define it as
\[
\Delta(q'_2,q'_1) = \left\{\bea{l}\displaystyle\frac{w(q'_2,q'_1)}{\proba_2(q'_2)} \mbox{ if } \proba_2(q'_2)>0 \\ [3mm] 0 \mbox{ otherwise.}
\eea\right.
\]
The conditions of Definition~\ref{dweiytdfutdrfyf} are satisfied. This finishes the proof of the ``if'' part and the lemma is proved. 
\end{proof}


\end{document}